\documentclass[letterpaper]{article} 
\usepackage{aaai18}  
\usepackage{times}
\usepackage{helvet}
\usepackage{courier}
\usepackage{url}
\usepackage{graphicx}
 
\usepackage[utf8]{inputenc}
\usepackage[T1]{fontenc}
\usepackage{amsmath}
\usepackage{amsfonts}
\usepackage{amssymb}
\usepackage{amsthm}
\usepackage{stmaryrd}
\usepackage{wasysym}
\usepackage{xspace}
\usepackage{booktabs} 
\usepackage{enumitem}
\usepackage{tikz}
\usetikzlibrary{positioning,fit,shapes}
\usepackage[%
colorinlistoftodos,
textsize=scriptsize,
obeyDraft
]{todonotes}
\usepackage{multicol}
\usepackage{thmtools}
\usepackage{thm-restate}
\usepackage{comment}
\usepackage{ifdraft}

\newtheorem{definition}{Definition} 

\newtheorem{lemma}[definition]{Lemma}
\newtheorem{proposition}[definition]{Proposition}
\newtheorem{corollary}[definition]{Corollary}
\newtheorem{example}[definition]{Example}
\newtheorem{fact}[definition]{Fact}

\usepackage{etoolbox}
\apptocmd{\sloppy}{\hbadness 10000\relax}{}{}

\newcommand{\rafael}[1]{\todo[inline,author=Rafael,color=cyan]{#1}}
\newcommand{\anni}[1]{\todo[inline,author=Anni,color=yellow]{#1}}
\newcommand{\veronika}[1]{\todo[inline,author=Veronika,color=lightgray]{#1}}

\newcommand{\newpart}[1]{\ifdraft{{\color{blue} #1}}{#1}\xspace}

\newcommand{\efont}[1]{\ensuremath{{\mathsf{#1}}}\xspace}

\newcommand{\exName}{1}
\newcommand{\Aex}{\ensuremath{\Amc_{\exName}}\xspace}
\newcommand{\Tex}{\ensuremath{\Tmc_{\exName}}\xspace}
\newcommand{\Kex}{\ensuremath{\Kmc_{\exName}}\xspace}

\newcommand{\IKex}{\ensuremath{\Imc_{\Kex}}\xspace}

\newcommand{\mypar}[1]{\medskip\noindent\textbf{#1.}}

\newcommand{\Con}{\ensuremath{\mathbb{C}}\xspace}

\newcommand{\Amc}{\ensuremath{\mathcal{A}}\xspace}

\newcommand{\Hmc}{\ensuremath{\mathcal{H}}\xspace}
\newcommand{\Imc}{\ensuremath{\mathcal{I}}\xspace}

\newcommand{\Kmc}{\ensuremath{\mathcal{K}}\xspace}

\newcommand{\Rmc}{\ensuremath{\mathcal{R}}\xspace}

\newcommand{\Tmc}{\ensuremath{\mathcal{T}}\xspace}
\newcommand{\Umc}{\ensuremath{\mathcal{U}}\xspace}

\newcommand{\Nbb}{\ensuremath{\mathbb{N}}\xspace}

\newcommand{\EL}{\ensuremath{\mathcal{E\!L}}\xspace}

\newcommand{\ELHbot}{\ensuremath{\mathcal{E\!LH}_\bot}\xspace}
\newcommand{\RELHbot}{\ensuremath{\mathcal{E\!LH}^{\rho}_{\bot}}\xspace}
\newcommand{\ALC}{\ensuremath{\mathcal{ALC}}\xspace}

\newcommand{\NC}{\ensuremath{\mathsf{N_C}}\xspace}
\newcommand{\NR}{\ensuremath{\mathsf{N_R}}\xspace}

\newcommand{\NI}{\ensuremath{\mathsf{N_I}}\xspace}
\newcommand{\NT}{\ensuremath{\mathsf{N_T}}\xspace}
\newcommand{\NIA}{\ensuremath{\mathsf{N}^{\Con}_{\mathsf{I}}}\xspace} 
\newcommand{\NIT}{\ensuremath{\mathsf{N}^\rho_{\mathsf{I}}}\xspace}
\newcommand{\NIU}{\ensuremath{\mathsf{N}^{\mathsf{up}}_{\mathsf{I}}}\xspace}
\newcommand{\NIL}{\ensuremath{\mathsf{N}^{\mathsf{low}}_{\mathsf{I}}}\xspace}
\newcommand{\NV}{\ensuremath{\mathsf{N_V}}\xspace}

\newcommand{\true}{\ensuremath{\mathsf{true}}\xspace}

\newcommand{\Cert}{\ensuremath{\mathsf{Cert}}\xspace}
\newcommand{\Ind}{\ensuremath{\mathsf{N_I}}\xspace}
\newcommand{\IndA}{\ensuremath{\mathsf{Ind}(\Amc)}\xspace}

\newcommand{\Var}{\ensuremath{\mathsf{N_V}}\xspace}
\newcommand{\Term}{\ensuremath{\mathsf{N_T}}\xspace}

\newcommand{\AVar}{\ensuremath{\mathsf{N_{AV}}}\xspace}
\newcommand{\QVar}{\ensuremath{\mathsf{N_{QV}}}\xspace}

\newcommand{\UO}{\ensuremath{{\Umc_\Kmc}}\xspace}
\newcommand{\IO}{\ensuremath{{\Imc_\Kmc}}\xspace}
\newcommand{\IOa}[1]{\ensuremath{{\Imc_{\Kmc_{ #1}}}}\xspace}
\newcommand{\IOR}{\ensuremath{{\Imc_\Kmc^{\mathsf{re}}}}\xspace}
\newcommand{\ir}{\ensuremath{\rho}\xspace}
\newcommand{\NRir}{\ensuremath{\NR\cup\{\ir\}}\xspace}
\newcommand{\ira}[1]{\ensuremath{\rho^{#1}}\xspace}
\newcommand{\iras}[2]{\ensuremath{\rho^{#1}_{#2}}\xspace}
\newcommand{\irO}{\ensuremath{\rho_{\Kmc}}\xspace}
\newcommand{\irl}{\ensuremath{\rho_{\ell}}\xspace}
\newcommand{\irla}[1]{\ensuremath{\irl^{#1}}\xspace}
\newcommand{\rhoior}{\ensuremath{\rho^{\IOR}}\xspace}

\newcommand{\aquery}{\ensuremath{\Phi}\xspace}

\newcommand{\erq}{{\ensuremath{\sim_{\Phi}^{\rho}}}\xspace}
\newcommand{\erqr}{{\ensuremath{\sim_{\Phi}^{r}}}\xspace}
\newcommand{\erqfiv}{{\ensuremath{\sim_{\Phi_5}^{\rho}}}\xspace}
\newcommand{\erqfivr}{{\ensuremath{\sim_{\Phi_5}^{r}}}\xspace}
\newcommand{\erqOp}{\;\erq\;}  
\newcommand{\erqrOp}{\;\erqr\;} 

\newcommand{\Aux}{\ensuremath{\mathsf{Aux}}\xspace}
\newcommand{\Paths}{\ensuremath{\mathsf{Paths}}\xspace}
\newcommand{\Tail}{\ensuremath{\mathsf{Tail}}\xspace}
\newcommand{\rhoTail}{\mbox{\ensuremath{\rho$-$\!\mathsf{Tail}}}\xspace}
\newcommand{\rhoTailFunc}[1]{\ensuremath{\rhoTail\big(#1\big)}\xspace}

\newcommand{\Pre}{\ensuremath{\mathsf{Pre}}\xspace}
\newcommand{\In}{\ensuremath{\mathsf{In}}\xspace}

\newcommand{\ForkId}{\ensuremath{\mathsf{Fork_=}}\xspace}
\newcommand{\ForkNeq}{\ensuremath{\mathsf{Fork_{\not=}}}\xspace}
\newcommand{\ForkH}{\ensuremath{\mathsf{Fork_\Hmc}}\xspace}
\newcommand{\I}{\ensuremath{\mathsf{I}}\xspace}

\newcommand{\Cyc}{\ensuremath{\mathsf{Cyc}}\xspace}

\newcommand{\Rew}{\ensuremath{\Phi^\dag_\Rmc}\xspace}

\let\oldDelta\Delta
\renewcommand{\Delta}{\mathrm{\oldDelta}}

\newcommand{\indHyp}{induction hypothesis\xspace}

\newcommand{\ourtitle}{\protect{%
	Query Answering for Rough \EL Ontologies
	}}

  \includecomment{tr}
  \excludecomment{paper}

\allowdisplaybreaks

\begin{tr}
	\title{\ourtitle\\(Extended Technical Report)}
\end{tr}
\begin{paper}
  \title{\ourtitle}
\end{paper}

\author{
	Rafael Pe\~naloza \\ KRDB Research Centre, \\Free University of Bolzano, Italy \\ \tt{\small Rafael.Penaloza@unibz.it}
	\And Veronika Thost \\ MIT-IBM Watson AI Lab\\IBM Research\\ \tt{\small veronika.thost@ibm.com}
	\And Anni-Yasmin Turhan \\ Inst.\ of Theor.\ Computer Science\\ TU~Dresden, Germany  \\ \tt{\small Anni-Yasmin.Turhan@tu-dresden.de}
}

\begin{document}
\maketitle{}

%

\begin{abstract}

Querying large datasets with incomplete and vague data is still a 
challenge. Ontology-based query answering extends standard database 
query answering by background knowledge from an ontology to augment 
incomplete data. We focus on ontologies written in rough description 
logics (DLs), which allow to represent vague knowledge by partitioning 
the domain of discourse into classes of indiscernible elements.

In this paper, we extend the combined approach for ontology-based query 
answering to a variant of the DL \ELHbot augmented with rough concept constructors. 
We show that this extension preserves the good computational properties 
of classical \EL and can be implemented by standard database systems. 
\end{abstract}

%
%

\section{Introduction}

Ontology-based query answering performs database-style query answering over description logic (DL) knowledge 
bases (KBs), 
which consist of an ontology (or TBox) expressing terminological (i.e., background) knowledge about a domain, and a 
dataset (called ABox) containing facts about particular individuals. The knowledge in the KB is 
captured by means of concepts (unary predicates) and roles (binary relations). The use  of conceptual background 
knowledge allows one to derive more answers to queries than from the data alone. The queries considered are 
typically conjunctive queries, which are special forms of first-order (FO) queries. 
The expressivity of a DL is determined by the concept (and sometimes also role) constructors it provides to describe important notions from the application domain. 
In classical DLs concepts represent unary predicates and hence are interpreted as sets of elements. Thus,
classical DLs lack capabilities of modeling uncertainty or vagueness \cite{LuSt08}. 

A moderate form of relaxation of concepts can be achieved by interpreting them as \emph{rough sets} \cite{Pawlak-PR-82}. 
Rough sets employ an \emph{indiscernibility relation} $\ir$, which groups objects that are considered to be 
indistinguishable from one another. The relation \ir effectively partitions the set of elements into so-called 
\emph{granules}. A granule, in essence, relaxes the notion of an element to a class of equivalent elements.
In rough sets, every classic set, say $S$, is accompanied by two sets. The \emph{lower approximation} $\underline{S}$ contains 
elements that all share the properties of elements in $S$ as it contains those partitions that lie completely in $S$. 
The \emph{upper approximation} $\overline{S}$  contains elements that are indistinguishable from an element in $S$, i.e., it 
contains those granules that  overlap with $S$.
%
Rough sets are employed in knowledge discovery and data mining, among others \cite{lin2012rough}.\anni{NTS: Expand.}

The capability of rough sets to relax objects in the data was already noticed in \cite{Pawlak-RoughReasoning-98} and 
is a standard way to relax database queries.
One of the goals of this paper is to extend these ideas to relax ontology-based query answering techniques.

In the context of DLs, concept constructors for upper (and lower) 
approximations provide means to relax (and crispen) concepts, while granules effectively relax objects. 
The idea to use rough set interpretations for DLs is not new 
\cite{liau1996rough,KleinMS-URSW-07,ScKP07,JiWaTX-09,Keet-EKAW10}. Rough DLs typically have concept constructors for the upper and  the lower approximation of concepts. 
%
One of their basic motivations is medical applications \cite{KleinMS-URSW-07,ScKP07}, where, 
for instance, patients can be indistinguishable by their symptoms or drugs and their generica can be 
indistinguishable by their active agent. Similarly, they were suggested to enhance the web ontology language 
OWL \cite{Keet-EKAW10} or to solve the identity matching problem in the linked data cloud 
\cite{KleinMS-URSW-07,BSvh-ESWC-16}. As in database settings, indiscernibility relations for rough DLs can 
be derived automatically from the data \cite{dAmFEL-13,BSvh-ESWC-16} making rough DLs amenable 
for practical applications.

Another approach for dealing with vagueness is based on fuzzy logic.
While fuzzy DLs \cite{BCE+-15} can express vagueness regarding the concept membership of objects, rough DLs can express 
granularity of objects. The former DLs can easily turn undecidable \cite{BoDP-AIJ15,BoCP-IJA17}, but the latter 
are always decidable, as long as the underlying classical DL is. 
Reasoning procedures for classical reasoning tasks such as satisfiability 
or subsumption, i.e.,\ the computation of  sub- and super-concept relationships in rough DLs were
proposed in
\cite{KleinMS-URSW-07,Keet-11,PenZ13}. In fact, if inverse roles, transitive roles and role hierarchies are available 
in a DL, then reasoning in its rough variant can be reduced to it \cite{KleinMS-URSW-07}.
%
The lightweight DL \EL has only conjunction and existential restrictions 
as concept constructors and thus such a reduction would use a much more expressive logic with higher 
computational complexity. \EL cannot express contradictions, thus subsumption is the interesting reasoning 
task, and can be decided in polynomial time \cite{BaaderBrandtLutz-IJCAI-05} by means of canonical models 
\cite{LuWo-JSC-10}. The  subsumption decision procedure  based on canonical models was lifted in 
\cite{PenZ13} to \RELHbot---a rough variant of \EL with role hierarchies extended by constructors for upper 
and lower approximations of concepts. 
\newpart{
This rough DL can be used, for example, to model biological species through their phenotypical characteristics,
which are often vague in nature. For example, the edible \emph{Agaricus arvensis} mushroom is described to have
an ``anise-like'' smell, ``ellipsoid'' spores, among other characteristics. Thus, we can say that this mushroom
belongs to the concept
\[
\efont{Edible} \sqcap \exists \efont{hasSmell}.\overline{\efont{Anise}}\sqcap 
	\exists \efont{hasSpores}.\exists \efont{hasShape}.\underline{\efont{Ellipse}}
\]
}


We consider ontology-based query answering in \RELHbot.   For this task, we use conjunctive queries 
that admit \RELHbot concepts and the indiscernibility relation \ir in the atoms of the query. 
\newpart{
For example, when preparing a field-guide to mushroom picking, it is important to highlight possible confusions 
between edible and poisonous mushrooms to avoid an intoxication. More precisely, one could query for all
pairs of mushrooms that are morphologically similar, but where one is edible and the other is not, through the query
\begin{align*}
\aquery(x_1,x_2) = \exists y_1,y_2.
	& \efont{Mushroom}(x_1)\land \efont{Edible}(x_1) \land {}\\
	& \efont{Mushroom}(x_2)\land \efont{Poisonous}(x_2) \land {}\\
	& \efont{hasShape}(x_1,y_1) \land {} \\
	& \efont{hasShape}(x_2,y_2) \land \ir(y_1,y_2).
\end{align*}
Such a query can be further refined, for example,  to return additionally the smell of the poisonous elements, or to 
consider other characteristics like color, size, or the shape of the spores. In this case, the query described
above could return the two answers that \emph{Agaricus arvensis} (which is edible) may be confused with the poisonous
\emph{Agaricus xanthodermus} and with \emph{Agaricus pilatianus}. The refined query would state that both 
poisonous species have a pungent smell, which makes them easy to differentiate from \emph{A. arvensis}.

Obviously, the relevance of rough CQ answering is not limited to the identification of mushrooms or other biological
species. It has also applications in medicine \cite{ScKP07}, for suggesting adequate treatments after identifying symptoms, and diseases, which usually have vague descriptions.
Furthermore rough CQ answering is applied in
verification, for quality control; and in online marketing, for handling similar clients uniformly, among many
others.
}


A well-known approach to answering conjunctive queries for classical \EL is the \emph{combined approach} \cite{LutzTW09}.  
It proceeds in two steps.
First, all the knowledge from the TBox is `absorbed' into the ABox. After this step only the data in the 
materialized ABox, but not the TBox, needs to be regarded for answering the query.
The materialized ABox introduces auxiliary elements to represent information about all syntactical sub-concepts 
occurring in the TBox. Hence, such a materialized ABox may give `spurious' answers to the original query, 
due to joins at auxiliary elements in the materialized ABox.
In the second step of the approach, the query is rewritten. The rewriting complements the query with filter conditions that sift out the spurious answers.
The combined approach is designed to be implemented by  database systems. The materialized ABox can be 
represented in a database and the rewritten conjunctive query can be expressed by standard database query 
languages. This approach has been implemented in competitive systems such as Combo system \cite{LSTW-ISWC-13}, 
and, based on Datalog, in RDFox  \cite{MNPHO-AAAI-14} and Hermit \cite{SM-AAAI-15}.

To lift the combined approach for \EL to the rough DL \RELHbot, the materialized ABox needs to be further 
augmented by new auxiliary elements. \newpart{These new elements represent the upper and lower 
approximations of concepts. Due to their semantics, they can give rise to new kinds of joins, which can in 
turn cause new kinds of spurious elements that are not detected by the filters employed for the classical \EL 
query answering method. Thus, it is important to provide new filter predicates for the rewritten query in the 
presence of rough information.}

\begin{tr}
	This technical report extends the original paper \cite{paper} by an appendix that contains the missing
	proofs.
	In detail, this report is structured as follows:	the next section introduces the basic notions for (rough) DLs and conjunctive query answering. Section \ref{sec:canonical} describes the absorption of TBox information into the ABox. Section \ref{sec:rewriting} develops the new filter conditions for the query rewriting. Section \ref{sec:assorted-facts} discusses possible extensions of the setting considered in the technical sections, before
	concluding with an outlook for future work.
	Appendix \ref{sec:appA} covers the proofs and additional definitions for Section \ref{sec:canonical} and the Appendix  \ref{sec:appB} does so for Section \ref{sec:rewriting}.	
\end{tr}
\begin{paper}
	For lack of space, full proofs are not included in this paper. \newpart{They can be found in the appendix uploaded with this submission.} 
	
\end{paper}

%

\section{Preliminaries}\label{sec:preliminaries}
We introduce the rough DL \RELHbot, that extends the classical DL  \ELHbot by an indiscernibility relation and by concept constructors 
for the {lower} and the {upper approximation}. Based on this, we define the problem of answering conjunctive queries that we consider.

\mypar{Syntax} Let \NC, \NR, and \NI be non-empty, pairwise disjoint sets of
\emph{concept names}, \emph{role names}, and \emph{individual names}, 
respectively, and let \ir be the \emph{indiscernibility relation}. 
  \RELHbot  \emph{concepts} are built inductively by the following syntax rule
(where $A\in\NC$ and $r\in\NR$):
  \begin{align*}
  	C &::=~A\,\mid\,\top\,\mid\,\bot\,\mid\, C\sqcap C\,\mid\, \exists r.C\,\mid\,
  	\overline{D}\,\mid\, \underline{D}.
  \end{align*}
Concepts of the form $\overline{C}$ (resp.\ $\underline{C}$) are  called the \emph{upper} (resp.\ \emph{lower}) \emph{approximation} of $C$.
  Let $A\in\NC$, $r, s\in\NR$, $a,b\in\NI$, and $C$ and $D$ be concepts.
  \emph{Axioms} are the following kinds of expressions:
  \emph{general concept inclusions} (GCIs) of the
  form $C\sqsubseteq D$, \emph{role inclusions} (RIs) of the form $r\sqsubseteq s$, and \emph{assertions} of the form $A(a)$, $r(a,b)$, or $\ir(a,b)$.
  A \emph{TBox} \Tmc is a finite set of GCIs and RIs, and an \emph{ABox} \Amc is a finite set of assertions.  Together, 
  they form a \emph{knowledge base} (KB) $\Kmc=(\Tmc,\Amc)$.
    
    Note that the indiscernibility relation \ir is not an
    element of the set of role names \NR and  does not occur in TBoxes explicitly, but it  
    can be used directly in ABoxes to state that two objects cannot be distinguished.  
    The relation \ir is the basis for the semantics of the upper and lower approximation.

    We denote the sets of all concept names, role names, individual names, and concepts (including syntactic sub-concepts) occurring in a set $X$ of expressions by $\NC(X)$, $\NR(X)$, $\NI(X)$, and  $\Con(X)$, respectively.
    
   \mypar{Semantics}  
     An \emph{interpretation} $\Imc=(\Delta^\Imc,\cdot^\Imc)$ consists of a non-empty set $\Delta^\Imc$, 
     called the \emph{domain} of \Imc, and an \emph{interpretation function} $\cdot^\Imc$, 
     which assigns to every $A\in\NC$ a
     	set $A^\Imc\subseteq\Delta^\Imc$, to every $r\in\NR$ a binary
     	relation $r^\Imc\subseteq\Delta^\Imc\times\Delta^\Imc$, 
     	to every $a\in\NI$ an element $a^\Imc\in\Delta^\Imc$ such that, for all $a,b\in\NI$, $a^\Imc\neq b^\Imc$ if 
	$a\neq b$ (unique name assumption),
     	and to \ir an equivalence relation \ira{\Imc} on $\Delta^\Imc$.

    Let $[x]_{\sim}$ denote the  \emph{equivalence class} of $x\in\Delta^\Imc$ under the relation $\sim$.
      The function $\cdot^\Imc$ is extended to complex concepts by setting $\top^\Imc:=\Delta^\Imc$,
     $\bot^\Imc:=\emptyset$, and 
     \begin{align*}
     (D\sqcap E)^\Imc{} &:=D^\Imc\cap E^\Imc \\
     (\exists r.D)^\Imc&{} :=\{x\in\Delta^\Imc\mid\exists y\in\Delta^\Imc,(x,y)\in r^\Imc,y\in D^\Imc\}\\
     \overline{D}^\Imc&{} :=\{x\in\Delta^\Imc\mid[x]_{\ira{\Imc}}\cap D^\Imc\not=\emptyset\}\\
     \underline{D}^\Imc&{} :=\{x\in\Delta^\Imc\mid[x]_{\ira{\Imc}}\subseteq D^\Imc\}.
     \end{align*}
The \emph{granule} of an element $x\in\Delta^\Imc$ is the equivalence class $[x]_{{\ira{\Imc}}}$ of 
elements indiscernible from $x$.  Intuitively, $\overline{D}$ relaxes $D$ to the union of all those granules 
with elements in $D$. Inversely, $\underline{D}$ strengthens $D^\Imc$ to those elements whose granule is 
fully contained in $D^\Imc$. \newpart{Observe that the lower approximation behaves to some extent like a 
value restriction from more expressive DLs in the sense that it refers to \emph{all elements of a granule}.} 
%
\begin{figure}[tb]
	\centering \small
	\begin{tikzpicture} [domain=0:14, scale=0.52, baseline=0, >=stealth,opacity=0.75]
	\definecolor{MyColor}{rgb}{0.83,0.83,0.83};
	\foreach \y in {1} {\foreach \x in {3,4,5,6,7,8} {\path[fill=MyColor] (\x,\y) rectangle ++(1,1); } }
	\foreach \y in {2} {\foreach \x in {2,3,4,5,6,7,8,9} {\path[fill=MyColor] (\x,\y) rectangle ++(1,1); } }
	\foreach \y in {3} {\foreach \x in {1,2,3,4,5,6,7,8,9} {\path[fill=MyColor] (\x,\y) rectangle ++(1,1); } }
	\foreach \y in {4} {\foreach \x in {1,2,3,4,5,6,7,8,9} {\path[fill=MyColor] (\x,\y) rectangle ++(1,1); } }
	\foreach \y in {5} {\foreach \x in {2,3,4,5,6,7,8,9} {\path[fill=MyColor] (\x,\y) rectangle ++(1,1); } }
	\foreach \y in {6} {\foreach \x in {3,4,5,6,7,8} {\path[fill=MyColor] (\x,\y) rectangle ++(1,1); } }
	\definecolor{MyColor}{rgb}{0.7,0.7,0.7};
	\foreach \y in {2} {\foreach \x in {4,5,6,7} {\path[fill=MyColor] (\x,\y) rectangle ++(1,1); } }
	\foreach \y in {3} {\foreach \x in {3,4,5,6,7,8} {\path[fill=MyColor] (\x,\y) rectangle ++(1,1); } }
	\foreach \y in {4} {\foreach \x in {3,4,5,6,7,8} {\path[fill=MyColor] (\x,\y) rectangle ++(1,1); } }
	\foreach \y in {5} {\foreach \x in {4,5,6,7} {\path[fill=MyColor] (\x,\y) rectangle ++(1,1); } }
	\filldraw [fill=none,draw=black,thick,label={above:{$C$}}] (5.89,3.90) ellipse (3.99 and 2.6);
	\draw[step=1cm,very thin,opacity=0.4,black!40] (0,0) grid (11,8); 
	\draw (10mm,65mm) node[above,opacity=1] {\Large $\Delta^\Imc$};
	\draw (62mm,31mm) node[above,opacity=1] {\large $\underline{C}^\Imc$};
	\draw (15.3mm,38mm) node[above,opacity=1] {\large $\overline{C}^\Imc$};
	\draw (62mm,2mm) node[above,opacity=1] {\large ${C}^\Imc$};
	\end{tikzpicture}
	\caption{Semantics of a concept (ellipse), its upper (light grey) and lower (dark grey) approximation.}
	\label{fig:sem-app}
\end{figure}
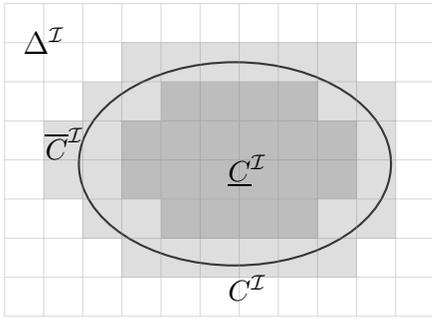
\newpart{The semantics of the upper approximation $\overline{C}$ and the lower approximation $\underline{C}$ are shown in Figure \ref{fig:sem-app} in relation to concept $C$. 
}    
     The interpretation \Imc is a \emph{model} of the GCI $C\sqsubseteq D$ iff
     $C^\Imc\subseteq D^\Imc$, the RI $r\sqsubseteq s$ iff
     $r^\Imc\subseteq s^\Imc$, the assertion $A(a)$ iff $a^\Imc\in A^\Imc$ 
     and the assertion $\widehat{r}(a,b)$ with $\widehat{r}\in\NRir$ iff  $(a^\Imc,b^\Imc)\in \widehat{r}^\Imc$.
An interpretation \Imc is a \emph{model} of (or \emph{satisfies}) a set of axioms $X$, written $\Imc\models X$, iff it is a model of all 
axioms in $X$.
A KB $\Kmc = (\Tmc, \Amc)$ is \emph{consistent} if $\Tmc \cup \Amc$ has a model, and \emph{inconsistent} otherwise. 
\Kmc \emph{entails} an axiom 
$\alpha$, written $\Kmc\models\alpha$, iff all models of \Kmc also satisfy $\alpha$. 
Given two concepts $C$ and $D$, we say that $C$ \emph{subsumes} $D$ w.r.t.\ \Kmc 
(written $C \sqsubseteq_{\Kmc} D$), 
iff  $C^\Imc \subseteq D^\Imc$ holds in every model \Imc of the KB \Kmc.

   \mypar{Query Answering} 
%
Consider a set of \emph{variables} \NV which is disjoint from $\NI \cup \NC \cup \NI$, and let $\NT:=\NV \cup \NI$ be the set of \emph{terms}. 
A \emph{first-order (FO) query} is a FO formula $\aquery(\vec{x})$ over the signature  $\NC\cup\NR\cup\{\ir\}\cup  \NT$.

The tuple
$\vec{x}=x_1,\ldots,x_k$ with $x_i\in\NV$ for all $i$, with $1\le i\le k$ are
the \emph{answer variables} of $\aquery(\vec{x})$. A query containing $k$ answer variables is a
$k$-ary query.
Let $C$ be an \RELHbot concept, ${r}\in\NR$, and $t,u\in\NT$. A \emph{conjunctive query} (CQ) is a FO query of the form $\exists \vec{v}.\aquery(\vec{v},\vec{w})$, where 
$\aquery$ is a (possibly empty) conjunction built of concept atoms $C(t)$, role atoms ${r}(t,u)$, and indiscernibility 
atoms $\ir(t,u)$. The empty conjunction is
denoted by \true.

Given an interpretation \Imc, a $k$-ary FO query $\aquery(\vec{x})$, and $a_i\in\NI$ for $i$, with $1\le i\le k$, 
we write $\Imc\models\aquery(a_1,\ldots,a_k)$ if the interpretation \Imc satisfies $\aquery(\vec{x})$ with $x_i$ assigned to $a_i^\Imc$ for 
$i$, with $1\le i\le k$, and call $(a_1,\ldots,a_k)$ an \emph{answer} to $\aquery$ in \Imc.
Such a tuple $(a_1,\ldots,a_k)$ 
is a \emph{certain answer} to $\aquery$ w.r.t.\  a KB \Kmc if, for every model \Imc of \Kmc, we have 
$\Imc\models\aquery(a_1,\ldots,a_k)$. The set $\Cert(\aquery,\Kmc)$ contains all
certain answers for a given CQ $\aquery$ w.r.t.\ a KB \Kmc. 
The reasoning task investigated in this paper is  \emph{CQ answering} in \RELHbot, i.e., the computation of the set $\Cert(\aquery,\Kmc)$. %

 When convenient, we view a CQ $\aquery$ as the set of atoms occurring in it.
For a given query \aquery we use the following sets: 
$\Term(\aquery)$ for its  terms,
 $\Var(\aquery)$ for its variables, 
$\AVar(\aquery)$ for its answer variables and $\QVar(\aquery)$ for its quantified variables.

%
\section{The Combined Approach for \RELHbot}

Recall that the combined approach for query answering first absorbs the TBox information into the ABox. 
Afterwards, it computes a query rewriting that augments the initial query by filter conditions. 

Let $\Kmc = (\Tmc, \Amc)$ be a KB.  For the remainder of the paper
we make the following simplifying assumptions w.l.o.g.\
\begin{enumerate}
	\item  CQs over \Kmc contain only individual names that occur  in \Kmc, 
	\item \Kmc contains no role synonyms; i.e., there are no $r,s\in\NR$ such that $r\not=s$ and $\Kmc\models
	\{r\sqsubseteq s, s\sqsubseteq r\}$, and 
	\item all concept names that appear in \Amc appear also in \Tmc.  
\end{enumerate}
For the rest of the paper let $\aquery$ be a $k$-ary CQ to be answered w.r.t.\ a consistent \RELHbot KB $\Kmc=(\Tmc,\Amc)$.

\subsection{Absorption of  TBox Axioms} 
\label{sec:canonical}

The goal of TBox absorption is to rewrite the ABox in such a way that all the background knowledge is
already included in it. In this way, the TBox can be disregarded in the
query answering process, using only the relevant information encoded in the rewritten ABox. 
We show how this method, originally devised for \EL, can be lifted to rough DL \RELHbot. 
\begin{figure*}[t]
	\fbox{%
	\parbox[t]{0.515\textwidth}{%
	\begin{align*}
	\\
	\Delta^\IO:={}&\Ind(\Amc)\cup\NIA\cup\NIL\cup\NIU
	\\[3mm] 
	a^\IO:={} & a
	\\
	\\
	\\
	r^\IO:={}& \{(a, b) \mid s(a,b)\in\Amc,~\Kmc\models s\sqsubseteq r\} ~\cup~
	\\
	&	\{(a,x_C)\in\Ind(\Amc)\times\NIA\mid \Kmc\models\exists r.C(a)\} ~\cup~
		\\
	&\{(x_C,x_D)\in\NIA\times\NIA\mid \Kmc\models
	C\sqsubseteq\exists r.D \} ~\cup 
	\\
	&	\{(x_{C,e},x_D)\in\NIU\times\NIA\mid \Kmc\models C\sqsubseteq\exists r.D \} ~\cup~
		\\
	&\{(x_{C,b},x_D),(\ell_{b},x_D)\in\NIT\times\NIA\mid \Kmc\models \underline{\exists r.D}(b)\} ~\cup
	\\
	&	\{(x_{C,x_E},x_D),(\ell_{x_E},x_D)\!\in\NIT\!\times\NIA\!\mid \Kmc\!\models E\sqsubseteq\underline{\exists r.D} \} \!\!\!
	\end{align*}
	}
	~~~
	\parbox[t]{0.445\textwidth}{
	\begin{align*}
	A^\IO:={}&
	\{a\in\Ind(\Amc)\ |\ \Kmc\models A(a)\} ~\cup~
	\\
	&\{x_C\in\NIA\ |\ \Kmc\models C\sqsubseteq A\} ~\cup~
	\\
	&\{x_{C,e}\in\NIU\ |\ \Kmc\models C\sqsubseteq A\} ~\cup~
	\\
	&\{x_{C,b},\ell_{b}\in\NIT\ |\ \Kmc\models\underline{A}(b)\} ~\cup~ 
	\\
	&	\{x_{C,x_D},\ell_{x_D}\in\NIT\ |\ \Kmc\models D\sqsubseteq \underline{A}\} 
	\\[2.5mm] 
	\iras{}{\Kmc}:={}& \{(a, b) \mid {\ir}(a, b)\in\Amc \} ~\cup~ 
	\\
	&	\{(a, x_{C,a})\in\Ind(\Amc)\times\NIU\mid \Kmc\models\overline{C}(a)\} ~\cup~ 
	\\
	&		\{(e,\ell_e)\mid\ell_e\in\NIL\} ~\cup{}
			\\
	&\{(x_C,x_{D,x_{C}})\in\NIA\times\NIU\mid \Kmc\models C\sqsubseteq\overline{D} \} ~\cup~ 
	\\
	&	\{(x_{C,e},x_{D,e})\in\NIU\times\NIU\mid \Kmc\models C\sqsubseteq\overline{D} \} 
	\\[1mm]
	\ira{\IO} :={} & 
		\text{ reflexive, symmetric, transitive closure of } 
		\iras{}{\Kmc}
	\end{align*}%
}}
	\caption{The canonical interpretation $\IO=(\Delta^{\IO},\cdot^{\IO})$ of \Kmc, where
		$a\in\NI(\Amc), A\in\NC(\Kmc),$ and $r\in\NR(\Kmc)$. 
		}
	\label{def:io}
\end{figure*}

\newpart{
ABox rewritings are usually represented as canonical interpretations. The canonical interpretations 
\cite{LuWo-JSC-10} used in the combined approach for \EL \cite{LutzTW09}, need to be extended for \RELHbot 
to accommodate the information from the upper and lower approximation concept constructors and from the 
$\rho$-assertions in the ABox. Canonical models that treat upper and lower approximations were previously
described in \cite{PenZ13}, where the goal was to decide concept subsumption and thus the focus was on the 
TBox only. For our case these canonical models need to be extended to represent the information from the (input) 
ABox too. 
} 

\newpart{
To formally define the canonical interpretations, we must introduce the \emph{normal form}. We say that 
a TBox is in normal form if all its GCIs are of the form
\begin{align*}
A\sqcap B & \sqsubseteq C, & 
\exists r.A & \sqsubseteq B,& 
A & \sqsubseteq \exists r.B,
\\
A & \sqsubseteq \underline{B}, & 
A & \sqsubseteq \overline{B},  & 
\underline{A}& \sqsubseteq B, 
\end{align*}
%
 where $A,B$ are concept names or $\top$ and $C$ is a concept name, $\bot$ or $\top$. 
Every \RELHbot TBox can be transformed to normal form in polynomial time. In the following we assume
that the TBox is always in normal form.
} 

The canonical interpretations of \RELHbot
contain four sorts of domain elements. We first give an overview of the sorts and then define the sets containing 
them. Two sorts are as in canonical interpretations for classical \EL: representatives for individual names 
occurring in the ABox \Amc\ \newpart{ (collected in the set $\NI(\Amc)$)}, and for concepts occurring in the 
TBox \Tmc \newpart{(collected in \NIA)}. 
We call these elements \emph{seed elements}. Additionally, we use two new sorts of domain elements: 
representatives for the lower approximations of each concept or individual occurring in the KB 
\newpart{(collected in \NIL)} and representatives for members of the upper approximations of concepts 
\newpart{(collected in \NIU)}.

We turn now to the definition of the sets capturing these four sorts of domain elements.
For simplicity, the \emph{named elements} representing the individual names are denoted by the corresponding 
names from $\NI(\Amc)$. 
The other elements are called
\emph{auxiliary elements} and are contained in the  sets:
\begin{align*}
\NIA:={}& 
\{x_C \mid C\in\Con(\Tmc) \} \\
\NIU:={}& 
\{x_{C,e}\mid C\in\Con(\Tmc),e\in\Ind(\Amc)\cup\NIA\} \\
\NIL:={}& 
\{\ell_{e} \mid e\in\Ind(\Amc)\cup\NIA
\}
\end{align*}
Intuitively, the auxiliary elements stand for the following:
\begin{itemize}
	\item $x_C\in\NIA$ represents an element that satisfies $C$ and acts as role-successor; 
		it is employed to make the predecessors satisfy concepts of the form  $\exists r.C$; 
	\item $x_{C,e} \in \NIU$ represents an element that 	
		satisfies $C$. \newpart{If the seed element $e$ is an individual, then $x_{C,e}$ is indiscernible from  $e$. In the case where the seed element $e$ is a concept, then $x_{C,e}$ represents that every element from $e$ is indistinguishable from some element in $C$. The element $x_{C,e}$} is used to make the seed element $e$ satisfy $\overline{C}$; and
	\item $\ell_e \in \NIL$ represents an element 
		satisfying exactly those concepts $C$ that are satisfied by all elements in the \newpart{lower approximation of $e$. If seed element $e$ is an individual, then}
		$\ell_e$ is indiscernible from element $e$. \newpart{If seed element $e$ is a concept, then $\ell_e$ represents all granules fully contained in  $e$.}
		The seed element $e$ satisfies $\underline{C}$ for all concepts $C$ associated to $\ell_e$.
\end{itemize}
\newpart{
	Sometimes we use the short-hand $\NIT = \NIU \cup \NIL$ for the `non-seed' elements.
	Observe that all elements in \NIU or \NIL are `caused' by a seed element. 
	The idea is that in the canonical interpretation each seed element is associated with an element representing this seed element's lower approximation. ABox individuals have the same granule as their lower or upper approximation, thus they only induce one element in \NIL. In contrast to this, concepts from \Tmc can have several granules in their approximations. Here, the lower approximation captures what is common to all granules in the lower approximation, thus one element in \NIL representing the lower approximation of a concept suffices. The granules in the upper approximation of a concept $C$ can overlap with different concepts or individuals $e$, thus different representatives for each such overlap are introduced in \NIU. 
During the reasoning process it can be discovered that some of the granule representatives belong into the 
same granule, which then gives rise to $\ir$-edges between the granule representatives.
}

The \emph{canonical interpretation} \IO of a KB $\Kmc$ is formally defined in Figure~\ref{def:io}, through a
description of the interpretation function of all the relevant elements.
The size of $\Delta^{\IO}$ is polynomial (more precisely, cubic) in the size of \Kmc. Moreover, \IO is computable in 
polynomial time, 
and consistency of \Kmc can be checked in polynomial time \cite{PenZ13}.
\anni{If time and space permits: elaborate on canonical model definition.}
\veronika{hmm. i am not sure. i would say it is pretty self-explanatory?}

\begin{example}
\label{ex:start}
Consider  $\Kex=(\Tex,\Aex)$ with
$\Tex=\{D\sqsubseteq \overline{C}, \\ C\sqsubseteq A\sqcap\underline{B} \}$, and 
$\Aex=\{ C(a), \overline{D}(a), \exists r.D(b), \ir(a,b) \}$. 
	
Figure~\ref{fig:elcmlowup} depicts its canonical interpretation \IKex (omitting transitive, reflexive $\rho$-edges).
As in  classical canonical interpretations, 
$a$ is an instance of $A$ since ($*$) \IKex satisfies both $C(a)$ and $C\sqsubseteq A\sqcap\underline{B}$. E
The element $b$ is an instance of $\exists r.D$, since it is related to the representative instance of $D$ 
($x_{D}\in\NIA$) 
via $r^{\IKex}$.

In the rough setting, the relation \ira\IKex comes into play and
($*$) yields that $a$ is an instance of $\underline{B}$; i.e., all elements in $[a]_\ir$, especially 
$\ell_a\in\NIL$, instantiate $B$ in \IKex. 
Since $D\sqsubseteq \overline{C}$, 
$x_{D}$ instantiates $\overline{C}$; i.e., it is related via \ira{\IKex} to its representative \ira{}-successor instantiating 
$C$: $x_{C,x_{D}}$.
The latter similarly holds for $x_{D, a}$, the representative \ira{}-successor of $a$ that instantiates $D$; 
i.e., $x_{D, a}$ is related via \ira{\IKex}$\!\!\!$ to its representative successor that is an instance of $C$, $x_{C,a}$.
Note, that $x_{D, a}$ exists due to the assertion $\overline{D}(a)$.
\end{example}
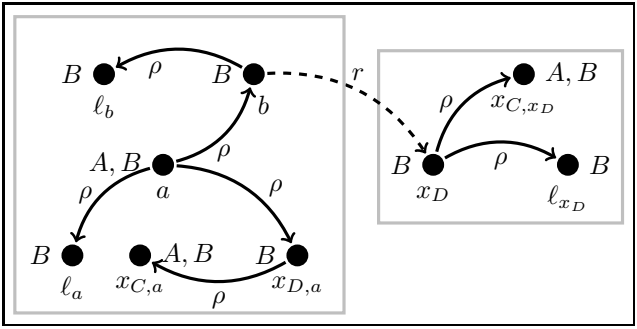
\begin{figure}[bt]
\hfil
\fbox{\parbox[t]{0.97\columnwidth}{\centering
	\begin{tikzpicture}[very thick,->,shorten >=1pt,shorten <=1pt] 
	\tikzstyle{dot} = [circle, draw,inner sep=2.5pt,fill]
	\node[dot,label=below:{$a$},label=left:{$A,B$}] (a) {};
	\node[dot,above right=of a,label=below:{$\ \ \ b$},label=left:{$B$}] (b) {};
	\node[dot,right=33mm of a,label=below:{$x_D$},label={[name=lxdb]left:{$B$}}] (xd) {};
	%
	%
	\node[dot,right=15mm of xd,label={[name=llxd]below:$\ell_{x_D}$},label={[name=lxdB]right:{$B$}}] (lxd) {}; 
	\node[dot,left=17mm of b,label=below:{$\ell_b$},label=left:{$B$}] (lb) {}; 
	\node[dot,below left=of a,label=below:{$\ell_a$},label=left:{$B$}] (la) {};
	%
	\node[dot,right=27mm of la,label=below:{$x_{D,a}$},label=left:{$B$}] (xda) {};
	\node[dot,left=18mm of xda,label=below:{$x_{C,a}$},label=right:{$A,B$}] (xca) {};
	\node[dot,above right=of xd,label=below:{$x_{C,x_D}$},label={[name=lxcxd]right:{$A,B$}}] (xcxd) {};
	%
	\node[fit=(xd)(lxd)(llxd)(xcxd)(lxdB)(lxcxd)(lxdb),rectangle,draw,color=lightgray!,inner sep=.2mm] (xdbox) {};
	\draw[bend right] (a) to node[below]{$\ir$} (b);
	\draw[bend left, style=dashed] (b) to node[above] {$r$} (xd);
	\draw[bend right] (a) to node[left] {\ir} (la);
	\draw[bend right] (b) to node[left] {\raisebox{-7mm}{\ir}} (lb);
	\draw[bend left] (xd) to node[below] {\ir} (lxd);
	%
	\draw[bend left] (a) to node[right] {$\ \ \ir$} (xda);
	\draw[bend left] (xda) to node[below] {\ir} (xca);
	\draw[bend left] (xd) to node[left] {\ir} (xcxd);
	\node[fit=(a)(b)(la)(lb)(xda.south)(xca),rectangle,draw,color=lightgray!,inner sep=6.2mm] (abox) {};
	\end{tikzpicture}
}} 
\caption{The canonical interpretation \IKex (without  transitive, reflexive {\ira{\IKex}\!\!\!}-edges and $C, D$ labels) 
for the KB $\Kex=(\Tex,\Aex)$ from Example~\ref{ex:start} shown as a graph.
Nodes represent domain elements and are labeled by the concept names they instantiate, 
edges represent relations. The gray frames highlight the granules of $a$ and $x_D$.  } 
\label{fig:elcmlowup}
\end{figure}
%
Figure \ref{fig:elcmlowup} shows that canonical interpretations in \RELHbot correspond to the ones in \EL modulo the granules---by regarding each granule as a single element, the result is an \EL interpretation that satisfies the TBox without approximation constructors.
\newpart{However, role assertions from the ABox can establish role edges between members of the same granule.}

\begin{lemma}
	\label{lem:iomodel}
	If \Kmc is consistent, then $\IO $ is a model of~\Kmc.
\end{lemma}
\begin{proof}[Proof (Sketch)]
	By construction, $\IO $ is a model of \Amc and of all RIs in
	\Tmc. We need to show that the GCIs in \Tmc are satisfied. 
	By induction on the concept structure it can be shown that, for all $C\in\Con(\Tmc)$, 
	$a \in\NI(\Amc)$, and $x_E\in\NIA$, it holds that
	$a\in C^{\IO }$ iff $\Kmc\models C(a)$, and 
	 $x_E\in C^{\IO }$ iff $\Kmc\models E\sqsubseteq C$.
	 Similar equivalences hold for elements of the form $x_{C,e}$ and $\ell_e$.
	 Then, it is easy to show that the GCIs $C\sqsubseteq D$ are satisfied by applying the corresponding 
	 equivalences to $C$ and~$D$.
\end{proof}
%
%
As mentioned already, the interpretation \IO can be seen as an ABox that encodes all the information stated in the 
original KB \Kmc.
However, queries cannot be answered using \IO directly, for two reasons. The first reason is, that 
the domain $\Delta^{\IO}$ of \IO may contain superfluous elements.
For example, for the KB $\Kmc_2=(\{C\sqsubseteq A\}, \emptyset)$, \IOa2 contains an element
$x_C\in A^{\IOa2}$\negmedspace. Thus, the CQ $\aquery_2=\exists y.A(y)$ would return an empty tuple (meaning that the query can be satisfied)
w.r.t.\
\IOa2, even though this is not an answer w.r.t.\ $\Kmc_2$.
We therefore restrict the canonical model \IO to those domain elements that are reachable from
named elements. 

A \emph{path} in an interpretation \Imc is a finite
sequence $d_0\widehat{r_1}d_1\cdots \widehat{r_n}d_n$, $n\geq 0$, such that
$d_0\in\Ind(\Amc)$ and, for all $i$ with $1\le i\le n$,
$d_i\in\Delta^{\Imc}\setminus\Ind(\Amc)$, 
$\widehat{r_i}\in\NR\cup\{\iras{}{\Kmc}\}$, 
and
$(d_{i-1},d_{i})\in
\widehat{r_{i}}^{\Imc}$. 
$\Paths(\Imc)$ denotes the set of all paths in \Imc.
For a path $p=d_0r_1d_1\cdots r_nd_n$,  
define
$\Tail(p):= d_n$. 
\newpart{
Intuitively, each path starts with an element that represents an ABox individual, each such element starts a 
path and there is no second ABox individual on a path. Observe that paths are defined using \iras{}{\Kmc} and 
not its symmetric, reflexive, transitive closure. 
} 

To avoid the superfluous domain elements, the interpretation $\IOR$ is obtained from \IO by restricting its domain elements to those reachable from elements that represent ABox individuals, or, more formally:
$$\Delta^{\IOR} = \{\Tail(p)\mid p\in\Paths(\IO) \}.$$
The next fact follows directly from this definition and states for those seed elements reachable by paths, the members of their granule. Thus it clarifies the picture of the indiscernibility relation in $\IOR$.
\begin{fact}
	\label{prop:iosim}
		For all seed elements that are reachable by paths, i.e., for all $e\in\NI(\Amc)\cup (\NIA\cap \Delta^{\IOR})$, we have
		\begin{align*}
		[e]_{\ira{\IOR}}~=~
		&\{e,\ell_{e}
		\}~\cup~
		\{ x_{D,e}\in\NIU\cap \Delta^{\IOR} \} ~ \cup{}\\
		&\bigcup_{\ir(e,a)\in\Amc}
		\Big(\{a,\ell_{a}
		\}\cup 
		\{ x_{D,a}\in\NIU\cap \Delta^{\IOR} \}\Big).\qed
		\end{align*}%
\end{fact}

The second reason why  queries cannot be answered using \IO directly, is the unintended reuse of some elements.
\newpart{
In the classical case of \EL, the elements in {\NIA} can introduce unintended joins in the model, and hence yield
erroneous answers. 
As noticed in \cite{LutzTW09}, for the KB $\Kmc_3= (\Tmc_3, \Amc_3)$ with 
$\Tmc_3= \{ A \sqsubseteq \exists r.B \sqcap \exists s.B\}$ and $\Amc_3 = \{A(a)\}$, the element $a$ is 
connected to $x_B$ via $r$ and $s$ in \IOa3. Considering the query 
$\aquery_3(x) = \exists y. r(x, y) \wedge s(x, y)$, this gives 
rise to $\IOa3 \models \aquery_3(a)$, but $a \not\in \Cert(\aquery_3, \Kmc_3)$. 

In the \RELHbot case, with the interpretation $\IOR$, the unintended reuse additionally  affects  those elements 
from \NIT 
(connected to the \NIA-elements) that were induced by seed elements from \NIA. So, for the KB $\Kmc_3$, there would be
 (among others) the element $x_{B, x_B} \in \NIU$ in the domain of $\IOa3$. This element is connected to element $x_B$ by a $\ir$-edge.
 For the query $\aquery_{3}'(x) = \exists y. r(x, y) \wedge s(x, y) \wedge \overline{B}(y) $ this gives 
rise to $\IOa3 \models \aquery_3'(a)$, but $a \not\in \Cert(\aquery_3', \Kmc_3)$. 
}  

To remedy these effects, the canonical model 
can be \emph{unraveled} into a new, tree-shaped interpretation \UO so that the paths in \IOR become the 
domain elements of \UO.
%
The \emph{unraveling} of \IOR is the interpretation
$\UO=(\Delta^{\UO},\cdot^{\UO})$, where, for all $a\in\NI(\Amc), A\in\NC(\Kmc),$  $r\in\NR(\Kmc)$: 
			\begin{align*}
			\Delta^{\UO}:={}& \Paths(\IOR)\qquad\qquad\qquad\qquad
			\\
			a^{\UO}:={} & a 
			\\
			A^{\UO}:={}&
			\{p \mid \Tail(p)\in A^{\IOR}\} 
			\\
			r^{\UO}:={}&
			\{(a, b)\mid a,b\in\Ind(\Amc), (a, b)\in r^{\IOR}\}\ \cup
			\\
			&\{(p,p\cdot se)\mid p,p\cdot se\in\Delta^{\UO},\Kmc\models
			s\sqsubseteq r\} 
			\\
			 \iras{}{\Kmc'}:={} &\{(a, b)\mid a,b\in\Ind(\Amc), (a, b)\in
			\ira{\IOR}\}~\cup 
			\\ 
			& 
			\{(p,p\cdot \ir e)\mid p\cdot\ir e\in\Delta^{\UO}\}
			\\
			\ira{\UO} :={} & \text{ reflexive, symmetric, transitive closure of } \iras{}{\Kmc'}
			\end{align*}
		Note that the construction of $\UO$ from $\IOR$ does not depend on the GCIs
		but only on the RIs in \Tmc. 
\begin{restatable}{lemma}{lemUOCert}
			\label{lem:uo-cert}
For every $a_1,\ldots,a_k\in\Ind(\Amc)$, we have  that\\
				{$(a_1,\ldots,a_k)\in\Cert(\aquery,\Kmc) \text{ ~iff~ }
				\UO\models\aquery[a_1,\ldots,a_k].$}	
\end{restatable}
%
The unraveling \UO gives the correct answers to CQs, but it is typically infinite; e.g.\ in the presence of terminological
cycles. The idea is therefore to focus on \IOR for 
CQ answering, but to take \UO as a kind of reference model. Specifically, the query $\aquery$ is extended with 
conditions that accept only answers compliant with \UO, by avoiding the unintended joins.


%
\subsection{The Query Rewriting}\label{sec:rewriting}

We focus now on the problem of \emph{rewriting} a CQ \aquery in such a way that the answers of its rewriting 
$\Rew$ w.r.t.\ \IOR correspond
exactly to the answers of the original query \aquery w.r.t.\ \Kmc.
More precisely, we want to prove the following result.
%
%
\begin{restatable}{theorem}{MainThm}
\label{thm:iocert} 
For every finite set of role inclusions \Rmc and each $k$-ary CQ $\aquery$,
one can construct in polynomial time a $k$\mbox{-}ary FO query
\Rew such that, for all \RELHbot KBs $\Kmc=(\Tmc,\Amc)$ 
using only the role inclusions \Rmc, and all 
$a_1,\ldots,a_k\in\Ind(\Amc)$,
we have  $$(a_1,\ldots,a_k)\in\Cert(\aquery,\Kmc)
\text{ ~iff~ } \IOR\models 
\Rew(a_1,\ldots,a_k).$$
\end{restatable}
\newpart{
In order to show this theorem, our first step is to develop the rewriting procedure. 
}
The combined approach extends a given CQ with additional
filter conditions to discard those answers to $\aquery$ in \IOR that are not answers in \UO. 
These conditions essentially target those parts of the CQ that can be satisfied by non-tree structures 
that may exist in \IOR but not in \UO. 
Observe that only non-tree structures including auxiliary elements are critical \newpart{as these are the ones
that would not appear in the original KB}. 
We extend the filter conditions from~\citeauthor{LutzTW09} to handle also the elements representing upper 
and lower approximations of concepts.
 
\newpart{Specifically, due to the properties of the indiscernibility relation $\rho$
(i.e., transitivity, reflexivity, and symmetry), and its influence in the approximation constructors, the new filter
conditions need to consider potential equivalences and joins that are only implicitly stated. For instance, a 
tree shaped query that leads to two different but indiscernible elements will include an implicit join that must
be taken into account.
}
%

Let \Rmc be an arbitrary but fixed finite set of RIs and $\aquery$ be a $k$-ary CQ. 
To identify  auxiliary elements, we introduce two fresh
unary predicates (that is, concepts): $\Aux$  identifies elements from \NIA and $\Aux_\ir$  
`approximation-related', i.e., `non-seed' elements from \NIT. 
We \newpart{define} them to be 
interpreted in \IOR and \UO 
as:
 \begin{align*}
 \Aux^{\IOR}&:=\Delta^{\IOR}\cap\NIA 
 \\
{\Aux_\ir}^{\IOR}&:=\Delta^{\IOR}\cap\NIT 
 \\ 
 \Aux^{\UO}&:=\bigcup_{p\in\Delta^{\UO},\Tail(p)\in\NIA}\{p\}
\\
{\Aux_\ir}^{\UO}&:=\bigcup_{p\in\Delta^{\UO},\Tail(p)\in\NIT}\{p\}
 \end{align*}
%
To model the filters, we describe those  mappings from answer variables to ABox individuals that describe non-tree structures  which cannot be satisfied in \UO. 
The latter is the case if the answer mapping uses a single \NIA element as a role successor for mapping several objects referred to in the query such that there is no corresponding element in \UO that fits all of them.
A corresponding such element in \UO exists, if the structures from the query can be mapped into a single path in \IOR, by identifying terms.

The terms that are identified in this way, and those that are indiscernible, are captured via an equivalence 
relation \erqr on terms, grouping them into equivalence classes. Let \erq be another equivalence relation over 
$\Term(\aquery)$ induced by the atoms of the form 
$\ir(s,t)$ occurring in $\aquery$ for some terms $s$ and $t$. 
The relation \erqr is defined inductively based on \erq as 
the smallest transitive and reflexive relation on $\Term(\aquery)$ that 
(1) includes 
 the relation
 $$\{(t,t')\mid r_1(s,t),r_2(s',t')\in\aquery,r_1,r_2\in\NR,t \erqOp t' \}$$
and (2) satisfies the closure condition:
	\begin{equation}\tag{$\dag$}
		\begin{split}
	\text{if } r_1(s,t),r_2(s',t')\in\aquery, \newpart{r_1,r_2\in\NR}
	\text{ and } t\erqrOp t',  \\ \text{ then } s\erqrOp s'. 
	\label{cond:closure}
		\end{split}
	\end{equation}
\newpart{Observe that the relation \erqr inherits symmetry by construction from the symmetric relation $\!\erqOp\!$ and, furthermore,  \erqr does not need to contain $\!\erqOp\!$ as a sub-relation.
}
The equivalence classes of \erqr group those terms that cannot be distinguished by homomorphisms from
$\aquery$ into $\UO$.
\newpart{
Such an inductively defined relation is already used in the combined approach for \EL \cite{LutzTW09}. The important difference is that in that previous work,} 
the induction is based on the identity relation. The closure condition then captures non-tree structures in the query $\aquery$, where a  term $t$ has two 
role-predecessors $s$ and $s'$.
\newpart{For \RELHbot, the identity relation is too fine-grained, since truly distinct objects belong to different granules. So, in order} 
to be able to handle in the query the relaxation introduced by the rough constructors, we need to consider 
the whole indiscernibility relation \newpart{on the query terms. Since granules can be separated by role 
relationships (as shown in Figure~\ref{fig:elcmlowup}), the incoming role edges of a granule and the related 
role-predecessor need to be addressed. In order to do so we define} for each equivalence class $\zeta$ of 
the relation \erqr \newpart{the predicates}: 
\begin{align*}
\Pre(\zeta)&:=\{ t\mid r(t,t')\in\aquery,r\in\NR, t'\in\zeta\}\\
\In(\zeta)&:=\{r\mid r(t,t')\in\aquery,
\newpart{
		r \in\NR 
		},
t'\in\zeta\}
\end{align*}
%
The set $\Pre(\zeta)$ describes all the role predecessors of terms in the equivalence class $\zeta$. The set $\In(\zeta)$ contains
all the incoming role names to $\zeta$.

\newpart{For the roles that separate the granules, the role hierarchy $\Rmc$ needs to be taken into account. As the more general role relationships of another is directly stated in the canonical model (by construction of $r^\IO$) and thus also in \IOR, the query needs to refer a most specific role.}
A role $r\in\NR$ is an \emph{implicant} of $R\subseteq\NR$ if $\Rmc\models r\sqsubseteq s$ for all 
$s\in R$. It is a \emph{prime implicant} if, additionally, $\Rmc\not\models r\sqsubseteq r'$ for all implicants $r'$ 
of $R$ with $r\neq r'$. Since KBs contains no role synonyms, there is a prime implicant for each $R\subseteq\NR$ for which there is an implicant. 

The different filters focus on different kinds of structures in $\aquery$. \newpart{ We collect these structures in the following sets, which are based on the sets $\Pre(\zeta)$ and $\In(\zeta)$, and on implicants:}
\begin{itemize}
\item 
\ForkNeq is the set of variables $v\in\QVar(\aquery)$ such that there is no implicant of $\In([v]_{\erqr})$.
Intuitively, \ForkNeq collects those variables that can never be mapped to the same \Aux-element in \UO, 
due to the shape of $\aquery$ (i.e., there are different role atoms where the variables occur as successors) and 
the interpretation of roles in \UO, which is based on the RIs entailed by \Kmc.
\item 
\ForkId is the set of pairs $(\Pre(\zeta), \zeta)$ with $|\Pre(\zeta)|\ge2$.
The first terms in the pairs in \ForkId are those variables that are mapped to indiscernible elements by any 
homomorphism of $\aquery$ into \UO and that may have to be identified if the successor variable is mapped to 
an \Aux-element.
Note that the case where the latter is not possible is captured by \ForkNeq.
Moreover, it does not suffice to require the identification, this is addressed next.
\item 
\ForkH is the set of pairs $(\I, \zeta)$ such that $\Pre(\zeta)\neq\emptyset$, there is a prime implicant of 
$\In(\zeta)$ that is not contained in $\In(\zeta)$, and \I is the set of all prime implicants of $\In(\zeta)$.
By the definition of \UO, a pair of an arbitrary element and an element of \NIA can be contained in the interpretations 
of different roles in \UO, but then it must also be in the interpretation of a prime implicant of those roles.
\ForkH therefore collects all relevant prime implicants so that the filter can enforce some such relation.
\item 
\Cyc is the set of all those quantified variables $v\in\QVar(\aquery)$ such that
there exist the role atoms $r_0(t_0,t'_0),\ldots,r_m(t_m,t'_m),\ldots,r_n(t_n,t'_n)$,
 $m,n\geq 0$ in $\aquery$ with 
$r_i\in\NR$
for all $i, 0\le i\le m$, and \newpart{the following conditions hold:}
	\begin{enumerate}
	\item  $(v,t_i)\in{\erqr\cup \erq}$ for some $i\leq n$, \item  $(t'_i,t_{i+1})\in{\erqr\cup\erq}$ for all $i<n$, 
	and 
	\item $(t'_n,t_m)\in \erqr\cup\erq$;
	\end{enumerate}
	i.e., \Cyc is the set of all \newpart{quantified} variables appearing in the query $\aquery$ that lead, through role connections and equivalences based on the indiscernibility relation,
	to cyclic dependencies. 
\end{itemize}
These definitions are analogous to those employed in the combined approach for \EL; 
the main change in our setting is the integration of the indiscernibility relation into $\erqr$ to capture the notion
of granules, \newpart{which is fundamental for the correctness of the method}.
\newpart{Notice that dealing with the indiscernibility relation requires a non-trivial extension of the 
classical case; indeed, indiscernible elements may affect many different points in the rewriting of a query.
Moreover, to keep the connection to the work by \citeauthor{LutzTW09} explicit, we have used the same
names for the filters; but they all differ from the original definitions.}

For each equivalence class $\zeta$ of \erqr, we select an arbitrary but fixed representative $t_\zeta\in\zeta$, 
and if $\Pre(\zeta)\neq\emptyset$, we also select a fixed element $t^\Pre_\zeta\in\Pre(\zeta)$.

Using these filters, we can now describe the promised query rewriting. Given the CQ $\aquery$, we define the
FO query
$$\Rew :=\exists\vec{x}.(\aquery'\wedge 
\Psi_1\wedge\Psi_2\wedge\Psi_3), \text{ where}$$ 
\begin{align*}
\Psi_1&:=\bigwedge_{v\in\AVar(\Phi)\cup\ForkNeq\cup\Cyc}\neg\Aux(v)\land \bigwedge_{v\in\AVar(\Phi)}\neg\Aux_\ir(v)\\
\Psi_2&:=\bigwedge_{(\{t_1,\ldots\,t_k\},\zeta)\in\ForkId}(\Aux(t_\zeta)\to \bigwedge_{1\leq i<k}t_i=t_{i+1})\\
\Psi_3&:=\bigwedge_{(\I,\zeta)\in\ForkH}(\Aux(t_\zeta)\to\bigvee_{r\in\I}r(t^\Pre_\zeta,t_\zeta)),
\end{align*}
and $\aquery'$ is a CQ  equivalent to $\aquery$ whose concept atoms are of the form $A(t)$ with $A\in\NC$.
This query $\aquery'$ it can be obtained from $\aquery$ through an unfolding that transforms complex concepts
into first-order terms. For example, the unfolding rewrites the conjunct $\overline{C}(x)$ in \aquery into
$\exists y.\ir(x,y)\wedge C(y)$.
%
Notice that the constraints enforcing that the explicit indiscernibility relations included in the original KB form
an equivalence relation are already encoded in the definition  of ${\Aux_\ir}^{\IOR}$ and ${\Aux_\ir}^{\UO}$.

The proof of Theorem~\ref{thm:iocert} focuses on the new query \Rew, \newpart{which can, in fact, be 
constructed in polynomial time. It remains to show that this query satisfies the property claimed by the
theorem}. The idea is that the filter conditions introduced
in the rewriting make sure that the answers over $\aquery$ that do not hold in \UO are excluded. 
$\Psi_1$ sifts out those answers in \IOR that contain auxiliary elements, and those that cannot be mirrored in \UO because the corresponding mapping uses an \NIA element as a role successor in several cases such that there is no corresponding element in \UO that fits all of them.
The query parts $\Psi_2$ and $\Psi_3$ characterize the situation in which a corresponding element in \UO exists: by identifying elements, the relevant structures from $\Phi'$ mapped into \IOR must be collapsible into a single path ($\Psi_2$), and a prime implicant must be among the edges between two nodes of this path ($\Psi_3$).

\newpart{
The proof of Theorem~\ref{thm:iocert} uses the FO query \Rew. The filtering conditions introduced
in this rewriting make sure that the answers over $\aquery$ that do not hold in the model \UO are excluded. 
For 
example, $\Psi_1$ guarantees, amongst others, that any cyclic dependency between domain elements 
must occur in the ABox. That is, cycles introduced by the reuse of auxiliary names in the canonical model are
ignored.
}
%
\begin{paper}
	\newpart{For lack of space, full proofs are not included in this paper. They can be found in the appendix uploaded with this submission. 
	}
\end{paper}
\begin{tr}
	The full proof is deferred to Appendix~\ref{sec:appB}.
\end{tr}

We now provide some simple examples of the rewriting, aimed to explain the ideas of the construction. 
Let $\Rmc=\emptyset$. Notice that in this case, \ForkH is always empty, and hence
$\Psi_3=\true$. We omit \Rmc and these $\Psi_3$ formulas in the rewritings.
%
We first demonstrate the role of \Cyc.
	Consider 
	$$\Phi_4:=\exists y_1,y_2.(\mathsf{hasA}(y_1,y_2)\land {\rho}(y_1,y_2)).$$
	%
	We have $\Cyc=\{y_1,y_2\}$, $\ForkId=\ForkNeq=\ForkH=\emptyset$,
	and thus obtain the following rewriting $\Phi_4^\dag$:
	$$
	\exists y_1,y_2.(
	\mathsf{hasA}(y_1,y_2)\land\rho(y_1,y_2)\land 
	\neg\Aux(y_1)\land\neg\Aux(x_2)).
	$$
	This query guarantees that all the answer pairs provided are indiscernible elements, related via the role $\mathsf{hasA}$, and that they contain no auxiliary elements.
%
We next consider a similar query, demonstrating the rewriting of forking situations:
	$$\Phi_5:=\exists y_1,y_2.(\mathsf{hasA}(x_1,y_1)\land \mathsf{hasA}(x_2,y_2) \land\rho(y_1,y_2)).$$
	The relation \erqfiv has equivalence classes $\{x_1\}$, $\{x_2\}$, and $\{y_1,y_2\}$, and 
	\erqfivr defines the partition $\{\{x_1,x_2\},\{y_1,y_2\}\}$.
	$\Pre(\{y_1,y_2\})=\{x_1,x_2\}$ and $\In(\{y_1,y_2\})=\{\mathsf{hasA}\}$.
	Thus, we have $\ForkId=\{(\{x_1,x_2\},\{y_1,y_2\})\}$, and $\ForkNeq=\ForkH=\Cyc=\emptyset$.
	This yields the rewriting
	\begin{align*}
	\Phi^\dag_5=&\exists y_1,y_2.(
	 \mathsf{hasA}(x_1,y_1)\land \mathsf{hasA}(x_2,y_2) \land\rho(y_1,y_2)\land{}\\
	& \neg\Aux(x_1)\land\neg\Aux(x_1)\land{} 
	 (\Aux(y_1)\to x_1=x_2)).
	\end{align*}

Notice that every step in the construction of the rewriting is polynomial in the size of the KB and the query.
Specifically, \erqr, \Pre, and \In are subsets of terms and variables that appear explicitly in \aquery. By extension,
the filters \ForkNeq, \ForkId, \ForkH, and \Cyc are also polynomial in \aquery. The only remaining case is 
ensuring that the auxiliary elements are not used to generate non-existing answers, as guaranteed by
the queries $\Psi_i, 1\le i\le 3$. The size of these queries is, in fact, polynomial in the number of auxiliary 
variables in \IOR. By construction, the domain of \IOR is polynomial in the size of \Kmc.
Overall, this means that the rewriting procedure runs in polynomial time, and produces a polynomially bounded
FO query.

%
\section{Reduction to Classical DLs}
\label{sec:assorted-facts}

After having considered the ontology-based query answering technique for rough DLs based on 
the combined approach in the last sections, 
we now take a brief look at a method for reducing this  problem to QA in 
classical DLs that builds on proposals developed for rough DLs in the past.

It is known that rough DLs can be simulated in sufficiently expressive (classical) DLs~\cite{ScKP07}.
Specifically, the upper and lower approximations $\underline{C}$ and $\overline{C}$ are equivalent to the 
concepts $\forall \rho.C$ and $\exists \rho.C$, respectively, where $\rho$ is a designated transitive, reflexive, and 
symmetric role. Hence, one needs only to be able to express existential and value restrictions (as in the DL
$\mathcal{ALC}$), and the three mentioned properties on roles. In other words, every rough-\EL KB can be expressed by an $\mathcal{SI}^{\sf Self}$ KB.%
\footnote{$\mathcal{SI}^{\sf Self}$ extends \ALC with transitive and inverse roles, and reflexivity statements. For
more details, see 
\cite{BCM+-07}.}
Thus, any QA tool capable of dealing with this (very) expressive DL would also be able to handle rough \EL. 
Given the efforts to produce efficient QA methods for expressive DLs, one obvious question is 
whether such methods can be exploited directly to handle \RELHbot. The answer, unfortunately, is `no'. 
The reason for this negative answer is that
this logic does not fall into the class of Horn DLs, for which QA tools are efficient. 
In a nutshell, Horn DLs are those that do not allow the 
expression of non-deterministic choices~\cite{OrRS11}.
 
\newpart{Recall the normal form for \RELHbot TBoxes presented at the beginning of the last section.}
It is easy to see that, under the translation of \citeauthor{ScKP07} described at the beginning of this section, 
all the axioms in the first row are in fact Horn axioms. Unfortunately, this does not hold for the last axiom since it
requires a value restriction on the left-hand side. This kind of constraint, which implicitly requires a 
non-deterministic choice (an element belongs to $\forall r.A$ if it either has no $r$-successors, or it has at least
one $r$-successor, and all of them belong to $A$), cannot be handled efficiently by state-of-the-art
QA tools.

On the other hand, the restriction of \RELHbot where lower approximations cannot appear on the left-hand side
of GCIs is, in fact, a Horn DL; more precisely, a sublanguage of Horn-$\mathcal{SI}^{\sf Self}$.
Obviously, this restriction removes an important part of the expressive power of roughness, which may be
fundamental for some practical applications. However, it is not hard to conceive cases where such
lower approximations on the left-hand side are not really necessary. For instance, in our species classification
and differentiation example, the TBox will fall within this sub-logic. Indeed, one may say that a property of
a species is indiscernible from another, but a meaningful species description will never say that an element
that is indiscernible from all in a species must satisfy some specific properties.


\newpart{There are approaches for conjunctive query answering that extend \EL directly towards the expressivity 
needed for rough \EL. For instance, in \cite{SM-AAAI-15} the authors investigate an extension of \EL 
that allows for reflexive and transitive roles, but not for symmetric ones, which in general damage the
tractability of \EL.  Their techniques were implemented in the system RDFox  \cite{MNPHO-AAAI-14}. 
As mentioned, this DL covers two of the three  properties of an equivalence relation. Symmetry for roles is missing 
in their approach, since symmetric roles behave to some extent similarly to inverse roles which are notorious for 
raising the computational complexity of reasoning in many logics. Even transitive roles alone are known to be a 
handicap to the performance of query answering systems for \EL including them. However, for \EL with transitive 
roles practical reasoning procedures based on the combined approach have been devised in \cite{LSTW-ISWC-13} 
and implemented in the Combo system. 
}


%

\section{Conclusions}

We have presented a combined approach for answering conjunctive queries in the rough DL \RELHbot.
\newpart{This approach first extends the input ABox to include also the knowledge encoded in the TBox by materialization}, and then rewrites
the query to guarantee that no answers are unexpectedly introduced in the first step. This allows us to effectively 
answer conjunctive queries in this rough DL using standard database technologies.

Interestingly, we have shown that dealing with this rough extension of \ELHbot does not incur in any increase
of complexity w.r.t.\ its classical counterpart; the rewriting remains polynomial in the size of the input.

Being able to model and reason with rough concepts is fundamental for applications in the life sciences, as they
allow the introduction of notions that cannot be precisely defined through use of approximating lower and upper bounds.
In addition, they allow to introduce examples of elements that cannot be distinguished by these approximations.
\newpart{Such approaches have recently been investigated for a more fine-grained setting, where vagueness can be captured by a similarity measure and a proto-typical instance, yielding a vague concept that can be dynamically relaxed or strengthened depending on a similarity threshold \cite{BBG-FROCOS-15}---albeit only for unfoldable TBoxes.
In our setting the query language itself allows to relax answers by admitting the indiscernibility relation and the approximation constructors in the query language. Here the degree of relaxation then depends on the presence of the indiscernibility relation in the data. A somewhat orthogonal approach has been investigated in \cite{EcPeTu-JAL15}, where the query language admits relaxation (of instance queries) by the use of a concept similarity measure and a threshold.
While the similarity-based approaches admit more flexibility, they crucially depend on the presence of an appropriate similarity measure supplied by the user. In case of approaches using rough DLs, the indiscernibility relation can, in principle, be automatically derived from the data \cite{dAmFEL-13,BSvh-ESWC-16}.
} 

We highlight that there exist database systems providing native support for 
rough sets \cite{HLH-FI-04,BB-LWA-15}. As an alternative approach, one could think of using them as a target
language for rewriting the queries. While this would solve some of the technical issues regarding indiscernible
elements in the query rewriting step, these systems are not as widely adopted and optimized as industrial
database systems. Hence we believe that our approach has a higher potential for practical impact.

We plan to implement the rewriting technique and to test its performance empirically. We will also 
extend our methods to weaker notions of roughness, by removing restrictions in the indiscernibility relation;
e.g.\ transitivity.

%


\section{Acknowledgements}

This work is partly supported by the German Research Foundation (DFG)
within the Cluster of Excellence ``Center for Advancing Electronics Dresden'' (cfaed) in CRC 912 (HAEC)  
and
within the project 
"Reasoning and Query Answering Using Concept Similarity Measures and Graded Membership Functions". \anni{In final, put number, Anni.}

{
	\bibliographystyle{aaai}
    \bibliography{short-string,ref}%
}

\begin{tr}
\clearpage
\appendix
The appendix provides proofs and additional definitions that were omitted from in
main text for lack of space.


\section{Proofs for Section~\ref{sec:canonical}}
\label{sec:appA}

We prove the following three claims, which establish the results from Section~\ref{sec:canonical}:
\begin{enumerate}[label=(A.\roman*)]
\item\label{item:A:one} \IOR is a model of \Kmc;
\item\label{item:A:two} \UO is a model of \Kmc; and
\item\label{item:A:three} the answers to $\Phi$ in \UO are the certain answers.
\end{enumerate}

Notice that~\ref{item:A:one} is similar to Lemma~\ref{lem:iomodel}, but the domain of the interpretation
is restricted to elements reachable from named individuals. To show this result, we prove the following
lemma, which classifies the instances of different concepts, according to their kind.
\begin{lemma}
	\label{lem:ioiffo}
	For all $C\in\Con(\Kmc)$; $a\in\Ind(\Amc); x_D\in\NIA;$ $ x_{E,a},$ $x_{E,x_D}\in\NIU;$ and $\ell_a,\ell_{x_D}\in\NIL$, the following hold:
	\begin{enumerate}[label=(\arabic*)]
		\item $a\in C^{\IOR }$ iff $\Kmc\models C(a)$.
		\item $x_D\in C^{\IOR }$ iff $\Kmc\models D\sqsubseteq C$.
		\item $x_{E,a}\in C^{\IOR }$ iff $\Kmc\models E\sqsubseteq C$ or
		$\Kmc\models\underline{C}(a)$.
		\item $x_{E,x_D}\in C^{\IOR }$ iff $\Kmc\models E\sqsubseteq C$ or
		$\Kmc\models D\sqsubseteq\underline{C}$.
		\item $\ell_{b}\in C^{\IOR }$ iff $\Kmc\models\underline{C}(b)$.
		\item $\ell_{x_D}\in C^{\IOR }$ iff $\Kmc\models
		D\sqsubseteq\underline{C}$.
	\end{enumerate}
\end{lemma}
\begin{proof}
	We prove the items simultaneously by induction on the structure of
	$C$.  The base case, where $C\in\NC$ is a direct consequence of the
	definition of \IOR.
	If $C$ is of the form $D\sqcap E$, the result follows trivially from the
	semantics and the induction hypothesis.
	We consider the remaining cases in detail.

	\noindent{\textbf{Case $C=\exists r.B$.}} (1) 
	($\Rightarrow$) If $a\in(\exists r.B)^{\IOR }$, then there is
	an $e\in\Delta^{\IOR }$ such that $(a,e)\in r^{\IOR }$ and $e\in
	B^{\IOR }$.  By the definition of $\IOR $, $e\not\in\NIT$.  If
	$e\in\IndA$, then $s(a,e)\in\Amc$ for some role $s$ with $\Kmc\models
	s\sqsubseteq r$, and the \indHyp yields  $\Kmc\models
	B(e)$; hence $\Kmc\models\exists r.B(a)$.  If $e$ is
	of the form $x_{D}\in\NIA$, then $\Kmc\models\exists r.D(a)$.
	Since 
the \indHyp further yields $\Kmc\models D\sqsubseteq B$, we get  $\Kmc\models\exists r.B(a)$.  
	{($\Leftarrow$)}
	If $\Kmc\models\exists r.B(a)$, then $(a,x_{B})\in r^{\IOR
	}$, by definition. The \indHyp also yields $x_{B}\in
	B^{\IOR }$. Hence, $a\in(\exists r.B)^{\IOR }$ follows.  
	The remaining sorts of domain elements can be treated analogously.
	
	\noindent{\textbf{Case $C=\overline{B}$.}} %
	($\Rightarrow$) 	(1) If $a\in\overline{B}^{\IOR }$,
	there is an $e\in\Delta^{\IOR }$ with $(a,e)\in{\ira{\IOR}}$
	and $e\in B^{\IOR }$. By Lemma~\ref{prop:iosim}, 
	(i)~$e\in\IndA$, (ii)~$e\in\NIT$ and has the form $x_{E,b}$, or (iii)~$\ell_b$ with
	$b\in\Ind(\Amc)$.  If (i), then $\Kmc\models B(e)$ by the
		\indHyp, and hence
	$\Kmc\models\overline{B}(a)$.
	If (ii), Lemma~\ref{prop:iosim} yields $(a,b)\in{\ira{\IOR}}$ and,	
	by the \indHyp,
	\[\Kmc\models E\sqsubseteq B, \tag{*}\label{tag:five}\] or $\Kmc\models\underline{B}(b)$.  
	In the latter case, the semantics directly yields
	$\Kmc\models\overline{B}(a)$ since $(a,e)\in{\ira{\IOR}}$. In
	the former case, the fact $x_{E,b}\in\Delta^{\IOR }$ together with the
	definition of \IOR implies $\Kmc\models\overline{E}(b)$. Thus,
	$\Kmc\models\overline{B}(b)$ by~\eqref{tag:five}. Thus,  $\Kmc\models\overline{B}(a)$.
	If (iii), Lemma~\ref{prop:iosim}
	yields $(a,b)\in{\ira{\IOR}}$, too.  By the \indHyp, it additionally holds that
	$\Kmc\models\underline{B}(b)$ and thus
	$\Kmc\models\overline{B}(a)$. 
	The proof for~(2) is very
	similar.  For~(3), we can restrict to the same sorts of
	elements $e$ as in the proof of~(1), by
	Lemma~\ref{prop:iosim}.  Then, $x_{E,a}\in\overline{B}^{\IOR
	}$ implies $a\in\overline{B}^{\IOR }$. By the \indHyp, we thus get
	$\Kmc\models\overline{B}(a)$, which corresponds to
	$\Kmc\models\underline{(\overline{B})}(a)$.  The proof of (5) is
	analogous to the one of (3), and the proofs of~(4) and~(6)
	similarly correspond to the one of (2).
	
	{($\Leftarrow$)}
	(1) If $\Kmc\models\overline{B}(a)$, then
	$(a,x_{B,a})\in{\ira{\IOR}}$.  From the \indHyp on $\Kmc\models B\sqsubseteq B$ yields $x_{B,a}\in B^{\IOR}$.  But then, the semantics directly yields
	$a\in\overline{B}^{\IOR }$.  
	The proof for~(2) is analogous.
	For (3), if $\Kmc\models E\sqsubseteq\overline{B}$ holds, the
	proof is analogous to the one of (1) and (2).  If
	$\Kmc\models\underline{(\overline{B})}(a)$, then	
	$\Kmc\models{(\overline{B})}(a)$.
 By definition of \IOR, we have $(a,x_{B,a})\in{\ira{\IOR}}$, and $x_{B,a}\in {B}^{\IOR}$ by the \indHyp. 
	$x_{E,a}\in\overline{B}^{\IOR }$ then follows from $(a,x_{E,a})\in{\ira{\IOR}}$, which must hold if $x_{E,a}
	\in{\Delta^{\IOR}}$.  
	The proof of~(4) is
	analogous, and the proofs of~(5) and (6) are analogous to the second
	cases in the proofs of (3) and (4), respectively.
	
	\noindent{\textbf{Case $C=\underline{B}$.}}
	($\Rightarrow$) (1) If $a\in\underline{B}^{\IOR }$, then all
	elements that are $\ira{\IOR }$-successors of $a$ must belong to $B$ in \IOR, too.
	By Fact~\ref{prop:iosim}, $(a,\ell_a)\in{\ira{\IOR}}$, and
	hence $\ell_a\in B^{\IOR}$. The \indHyp and the semantics then directly lead to
	$\Kmc\models\underline{B}(a)$.  The proofs for the other sorts of elements are  analogous.
	
	\noindent{($\Leftarrow$)}
	(1) We prove this case by contradiction. 
	 Suppose that $\Kmc\models\underline{B}(a)$ and that there is
	an element $e\in\Delta^{\IOR}$ such that $(a, e) \in \rhoior$ and
	$e\not\in B^{\IOR}$. By Fact~\ref{prop:iosim}, $e$ is
	either (i)~an individual name or from \NIT and of the form (ii) $x_{E,b}$ or $\ell_b$ with $E\in\Con(\Tmc)$ and $b\in\NI$; note that $a=b$ is possible.
	In the case~(i), we have
	$\ir(a,b)\in\Amc$, by Fact~\ref{prop:iosim}, and hence get
	$\Kmc\models B(b)$, by the semantics of the lower approximation. 
	But then, the application of the \indHyp yields $e=b\in B^{\IOR}$, which is a contradiction.  
	In case~(ii), $(a, e) \in \rhoior$ and Fact~\ref{prop:iosim} imply $b\in[a]_{\ira{\IOR}}$ and in particular 
	$\ir(a,b)\in\Amc$. Given
		$\Kmc\models\underline{B}(a)$, the semantics yields
		$\Kmc\models\underline{B}(b)$ which contradicts 
		$\Kmc\not\models\underline{B}(b)$. The latter follows from 
		$e\not\in B^{\IOR}$ by the \indHyp.
	For (2), we proceed similarly.
	Suppose that $\Kmc\models E\sqsubseteq \underline{B}$ and that there is an element $e\in\Delta^{\IOR}$ such that $(x_D, e) \in \rhoior$ and
		$e\not\in B^{\IOR}$. By Fact~\ref{prop:iosim}, $e$ is of the form  $x_{E,x_D}$ with $E\in\Con(\Tmc)$ or $\ell_{x_D}$.
		In both cases, the \indHyp directly implies the contradiction $\Kmc\not\models D\sqsubseteq\underline{B}$.
	For (3), there are two cases to be considered. However,
	given an element $x_{E,a}\in\Delta^{\IOR}$ (i.e., it is reachable in \IOR), the
	definition of \IOR together with the \indHyp (regarding $E$) yields that $\Kmc\models\overline{E}(a)$.  
	But, then, the first case, $\Kmc\models E\sqsubseteq\underline{B}$, by
	the semantics, implies the second case, 
	$\Kmc\models\underline{B}(a)$. That case can be treated as (1) since Fact~\ref{prop:iosim} yields the same structure of the equivalence class.
	Also (6) is treated in that way.
For (4), we again only have to consider the second case, as with (3), and it can be treated analogous to (2). The same holds for (6).
\end{proof}
Given Lemma~\ref{lem:ioiffo}, $\IOR \models \Kmc$ \ref{item:A:one} follows by the same arguments as applied for \IO in the proof of Lemma~\ref{lem:iomodel}.
We now proceed to show that \UO is a model of \Kmc with the help of~\ref{item:A:one}. 
To this end, we relate the interpretations \UO and \IOR to each other based on the correspondences 
between their domain elements. 
Recall that all elements in $\Delta^{\UO}$ are paths in \IOR.

We first provide results on the different kinds of domain elements in \IOR regarding their role as tails of the paths in \UO; each sort enforces the corresponding paths to be of a certain shape.
\begin{lemma}\label{prop:uoels}
	For all $d_0\hat{r}_1d_1\cdots \hat{r}_nd_n\in\Delta^{\UO}$, we have:
	\begin{enumerate}[label=(\arabic*)] 
		\item\label{prop:uoels:ind} 
		$d_n\in\Ind(\Amc)$ iff $n=0$.
		\item\label{prop:uoels:nia} $d_n\in\NIA$ iff 
		$\hat{r}_n\in\NR$.
			\item\label{prop:uoels:indniu}\label{prop:uoels:nianiu}
			$d_n\in\NIU$ iff there exists an $e\in\NI(\Amc)\cup\NIA$
				and an $i,0\le i<n$ such that
				$d_{i}=e$ and,  for all $j,i<j\le n$, $\hat{r}_{j}={\ir}$ and  
				$d_{j}=x_{C_{j},e}\in\NIU$ with $C_{j}\in\Con(\Tmc)$.
		\item\label{prop:uoels:nil} 
				$d_n\in\NIL$ 
				iff there is an $e\in\NI(\Amc)\cup\NIA$ such that
				$d_{n-1}=e$, $\hat{r}_{n-1}={\ir}$, and $d_n=\ell_e$.
	\end{enumerate}
\end{lemma}
\begin{proof}

	\ref{prop:uoels:ind} is a direct consequence of the definition of a path.
	
	\ref{prop:uoels:nia} follows from the definition of a path: 
	$d_n\in\NIA\cup\NIT$ and $\hat{r}_n\in\NR\cup\{\ir\}$; and by the definition of \IOR, 
	($\Rightarrow$) an element of \NIA cannot be a $\iras{}{\Kmc}$-successor 
	and ($\Leftarrow$) an element of \NIT cannot be a role-successor.
	
	\ref{prop:uoels:nianiu} and \ref{prop:uoels:nil} similarly follow from the definitions of a path and \IOR. 
	Regarding the latter and ($\Rightarrow$), an element of \NIT can neither be a role-successor nor a a $\iras{}{\Kmc}$-successor for an element of \NIL, and an element of \NIA cannot be a role-predecessor if it is not the one corresponding seed element. The directions ($\Leftarrow$) are trivial.
\end{proof}

The following corollary concretizes \UO even further, 
regarding the elements of $\Delta^{\UO}$ that are indiscernible. 
It directly follows from the definition of $\ira{\UO}$ based on the paths in $\Delta^{\UO}$ and 
Lemma~\ref{prop:uoels}, which specifies the latter.
\begin{corollary}
	\label{lem:uosimniaels}
	Suppose that $p\in[qr x_C]_{\ira{\UO}}$ with $x_C\in\NIA$. Then
	$p=q r x_C$,
	$p=q r x_C\ir\ell_{x_C}$, or 
	$p=qrx_C(\ir x_{D_i,x_C})^i$, $i\ge 0$. \qed
\end{corollary}
In order to relate the interpretation of \ir in \IOR to the one in \UO,
we show that the following properties hold: 
\begin{description}
	\item{(P1)} For each pair $ (p,q) \in \rho^{\UO}$, there is a corresponding tuple $(\Tail(p),\Tail(q)) \in \rho^\IOR$.
	
	\item{(P2)} For each pair $(d, e) \in \rhoior$, all ``copies'' of $d$ in \UO (i.e., all elements denoted by paths ending on $d$) have a $\rho$-successor in \UO.
\end{description}
Since $\rhoior$ and $\rho^{\UO}$ are obtained by symmetric, transitive, reflexive closures, these properties are not obvious. To show them, we define a function $\rhoTail \colon {\ira{\UO}}\to{\ira{\IOR}}$ as follows:
$$
 \rhoTailFunc{(p, q)}:= \big(\Tail(p), \Tail(q)\big).$$ 
We show that this function is well-defined to obtain P1 and that it is surjective to obtain P2. 
Note that \rhoTail is typically not a bijection, since $\Tail$ does not need to be injective. 
\begin{lemma}
	\label{lem:tails}
	Let $p,q\in\Delta^{\UO}$.
	If $(p,q)\in\ira{\UO}$, then  $\rhoTailFunc{(p,q)}\in\ira{\IOR}$.
\end{lemma}
\begin{proof}
	We prove this claim by induction on 
	$$k = min\big\{n \in \Nbb_0 \mid \{(p, q), (q, p)\} \cap ({\iras{}{\Kmc'}}\cup {\iras{-}{\Kmc'}})^n \not= \emptyset\big\},$$ 
	i.e., on the length of the shortest path between $p$ and $q$ in \UO consisting only of 
	$\iras{}{\Kmc'}$-edges or their inverses.	
	
	\noindent{\textbf{Case: $k=0$.}} In this case, we regard tuples $(p, p)$ in  $\iras{}{\Kmc'}$, which are also contained in $\ira{\UO}$. 
	Since $p\in \Delta^{\UO}$, $\Delta^{\UO}=\Paths(\IOR)$, and \rhoior is reflexive,
	$(\Tail(p), \Tail(p)) \in \rhoior$ holds.
	
	\noindent{\textbf{Case: $k=1$.}} In this case, $(p, q) \in \iras{}{\Kmc'}\cup\iras{-}{\Kmc'}$.
	If $\Tail(p),\Tail(q)\in\Ind(\Amc)$, then $(\Tail(p),\Tail(q))\in\ira{\IOR}$ holds by the definition of $\iras{}{\Kmc'}$. 
	Otherwise, we have $q=p\cdot\ir q'$, and $\Delta^{\UO}=\Paths(\IOR)$
	yields $(\Tail(p),\Tail(q))\in \ira{\IOR}$.
	
	\noindent{\textbf{Case: $k>1$.}} Then $(p,q)\not \in\iras{}{\Kmc'}\cup\iras{-}{\Kmc'}$ and $(p,q)$ is added  to 
	$\ira{\UO}$ by the transitive closure of $\iras{}{\Kmc'}$. Thus, there exists an element 
	$p'\in \Delta^{\UO}$ with $\{(p, p'), (p', q)\} \subseteq \ira{\UO}$. 
	Applying the \indHyp to this pair then yields:
	$\{(\Tail(p), \Tail(p')),\, (\Tail(p'), \Tail(q))\} \subseteq \rhoior$. From the transitivity of  \rhoior, 
	$\{(\Tail(p), \Tail(q))\} \in \rhoior$ follows.
	
	Since $\ira{\IOR}$ is the transitive, reflexive and symmetric closure of $\iras{}{\Kmc'}$, every pair of elements
	related via $\ira{\IOR}$ falls into one of the three cases above.
\end{proof}
The next lemma establishes surjectivity of the function  \rhoTail and thus the property P2 presented before.
\begin{lemma}
	\label{lem:tail:surj} %
	If $(d_n,e)\in\ira{\IOR}$, then for each $p\in\Delta^{\UO}$ with $\Tail(p)=d_n$ 
	there is
	an element $q\in\Delta^{\UO}$ with $\rhoTailFunc{(p, q)} = (d_n, e)$.
\end{lemma}
\begin{proof}
	Note that the fact that $\rhoTail$ is defined for $(p, q)$ implies that $(p, q)\in\ira{\UO}$.
	The lemma is shown by induction on $$k = min\big\{m  \in \Nbb_0 \mid \{(d_n, e), (e, d_n)\} \cap ({\iras{}{\Kmc}})^m \not= \emptyset\big\},$$ i.e.,  the length of the shortest path between $d_n$ and $e$ consisting only of $\iras{}{\Kmc}$-edges.
	
	\noindent{\textbf{Case $k = 0$.}} In this case, $d_n = e$ and 
	$(d_n, d_n) \in \iras{}{\Kmc} \subseteq \rhoior$ and thus 
	if $p \in \Delta^{\UO}$ and $\Tail(p) = d_n$, then
	there exists $q = p\cdot\rho d_n \in \Delta^{\UO}$ and $(p, q)\in \iras{}{\Kmc'} \subseteq \ira{\UO}$, which yields $\rhoTailFunc{(p, q)} = (d_n, d_n)$.
	
	\noindent{\textbf{Case $k = 1$.}} Then $(d_n,e)\in \iras{}{\Kmc}$ or $(e, d_n)\in \iras{}{\Kmc}$. W.l.o.g.\ assume that $(d_n,e)\in \iras{}{\Kmc}$.  
	If $d_n,e\in\Ind(\Amc)$, then  $(d_n,e)\in\iras{}{\Kmc'}\subseteq\ira{\UO}$ holds by
	definition of $\UO$ and thus $\rhoTailFunc{(p, q)} = (d_n, e)$.  
	Otherwise, if $e\not\in\Ind(\Amc)$ and if $p \in \Delta^{\UO}$ with $\Tail(p) = d_n$, then since
	$(d_n,e)\in\iras{}{\Kmc}\subseteq \ira{\IOR}$, there exists the
	element $q=p\cdot\ir e\in\Delta^{\UO}$, and hence 
	$(p,q)\in\iras{}{\Kmc'}\subseteq\ira{\UO}$, which yields $\rhoTailFunc{(p, q)} = (d_n, e)$.
	
	\noindent{\textbf{Case $k > 1$.}} Then, 
	$\{(d_n, e), (e, d_n)\} \cap ({\iras{}{\Kmc}})^k \not= \emptyset\big\}$. Assume w.l.o.g.\ 
	that $(d_n, e) \in  ({\iras{}{\Kmc}})^k$. This implies that there exists $f \in \Delta^{\IOR}$ such that 
	$\{(d_n, f), (f, e)\} \in ({\iras{}{\Kmc}})^{(k-1)} \subseteq \rhoior$. 
	If $p \in \Delta^{\UO}$ with $\Tail(p) = d_n$, then the  \indHyp implies that there exists an element 
	$p_f \in \Delta^{\UO}$ such that
	$\rhoTailFunc{(p, p_f)} = (d_n, f)$. In this case, the \indHyp also yields $q \in \Delta^{\UO}$ such that
	$\rhoTailFunc{(p_f, q)} = (f, e)$. This  implies that $\rhoTailFunc{(p, q)} = (d_n, e)$.
\end{proof}

Using these results, we can finally show that concept memberships coincide in \UO and \IOR. 
\begin{lemma}
	\label{lem:ioiffuo}
	For all $p\in\Delta^\UO$ and all $C\in\Con(\Kmc)$, we have 
	$p\in C^{\UO}$ iff $\Tail(p)\in C^{\IOR}$. 
\end{lemma}
\begin{proof}
	The claim is shown by induction on the structure of $C$.  If 
	$C\in\NC$, it follows from the
	definition of $\UO$.  
	The case $C=D\sqcap E$ also follows easily from the application of the \indHyp.
	
	\noindent{\textbf{Case $C=\exists r.D$.}}
	($\Rightarrow$) If $p\in (\exists r.D)^{\UO}$, then
	there exists a $q\in\Delta^{\UO}$ such that $(p,q)\in
	r^{\UO}$ and $q\in D^{\UO}$.
	By the definition of $\UO$, either (i) $p,q\in\Ind(\Amc)$, meaning $p=\Tail(p)$ and $q=\Tail(q)$, and
	$(p,q)\in r^{\IOR}$; or (ii) $q$ is of the form $q=p\cdot se$ with $s\in \NR$ and $\Kmc\models s\sqsubseteq r$.
	For the latter, $p\cdot se\in\Paths(\IOR)$ holds, by the
	definition of $\UO$, which implies $(\Tail(p),e)\in s^{\IOR}$, and $(\Tail(p),e)\in r^{\IOR}$ by Lemma~\ref{lem:iomodel}.  By the
	\indHyp, $\Tail(q)\in D^{\IOR}$ holds in both cases, and
	$\Tail(p) \in (\exists r.D)^{\IOR}$ follows.
	($\Leftarrow$) If $\Tail(p)\in (\exists r.D)^{\IOR}$, then
	there is an $e\in D^{\IOR}$ with $(\Tail(p),e)\in
	r^{\IOR}$.  By the definition of \IOR, either
	$e\in\Ind(\Amc)\subseteq\Delta^{\UO}$ and we set $q:=e$, or $q:=p\cdot re\in\Delta^{\UO}$.
	In both cases, the definition of $\UO$
	yields $(p,q)\in r^{\UO}$.  By the \indHyp, $q\in D^{\UO}$,
	and $p\in (\exists r.D)^{\UO}$ follows.
	
	\noindent{\textbf{Case $C=\overline{D}$.}}
	($\Rightarrow$) Let $p\in\overline{D}^{\UO}$, then
	there is some $q\in\Delta^{\UO}$
	such that $(p,q)\in\ira{\UO}$ and $q\in D^{\UO}$.  By
	Lemma~\ref{lem:tails}, we get $(\Tail(p),\Tail(q))\in\ira{\IOR}$, and
	by induction $\Tail(q)\in D^{\IOR}$.  Hence,
	$\Tail(p)\in\overline{D}^{\IOR}$.
	($\Leftarrow$) If $\Tail(p)\in\overline{D}^{\IOR}$, then
	there is some $e\in D^{\IOR}$ with $(\Tail(p),e)\in\ira{\IOR}$.  
	By Lemma~\ref{lem:tail:surj}, there is a
	$q\in\Delta^{\UO}$ such that $(p,q)\in\ira{\UO}$ and $\Tail(q)=e$.
	By the \indHyp, $q\in D^{\UO}$, and thus $p\in\overline{D}^{\UO}$.
	
	\noindent{\textbf{Case $C=\underline{D}$.}}
	($\Rightarrow$) If $p\in\underline{D}^{\UO}$, then
	$q\in D^{\UO}$ for all $q\in[p]_{\ira{\UO}}$. 
	By Lemma~\ref{lem:tail:surj}, for every $e\in[\Tail(p)]_{\ira{\IOR}}$,
	there is a $q\in [p]_{\ira{\UO}}$ with $\Tail(q)=e$. 
	Hence, $e\in  D^{\IOR}$ follows from the \indHyp. 
	This implies $\Tail(p)\in\underline{D}^{\IOR}$.
	($\Leftarrow$) If $d\in\underline{D}^{\IOR}$, then
	$e\in D^{\IOR}$ for all $e\in[d]_{\ira{\IOR}}$.  
	Let $q\in[p]_{\ira{\UO}}$ with $\Tail(p)=d$. By Lemma~\ref{lem:tails}, we have
	$\Tail(q)\in[d]_{\ira{\IOR}}$. By induction, $q\in D^{\UO}$, and
	hence $p\in\underline{D}^{\UO}$.
\end{proof}

It is now straightforward to establish the following result.
\begin{lemma}
	\label{lem:uomodel}
	$\UO$ is a model of \Kmc.
\end{lemma}
\begin{proof}
$\UO\models\Amc$ follows from Lemma~\ref{lem:iomodel}, from the fact that the domains of \UO and \IOR coincide on the named elements, from Lemma~\ref{lem:ioiffuo} regarding concept assertions, and from the definition of \UO based on that of \IOR regarding the remaining assertions.
The RIs in \Tmc are satisfied by the definition of \UO,
and the GCIs by Lemmas~\ref{lem:iomodel} and~\ref{lem:ioiffuo}.
\end{proof}

In the remainder of this section, we prove that \UO can be used for CQ answering, which establishes the 
claim~\ref{item:A:three}.
\lemUOCert*
\begin{proof}
($\Rightarrow$) This direction follows from Lemma~\ref{lem:uomodel}.

\noindent($\Leftarrow$) 
	Assume that $\UO\models\psi[a_1,\ldots,a_k]$ holds and let
	\Imc be an arbitrary model of \Kmc.  We define a mapping
	$\pi:\Delta^{\UO}\rightarrow\Delta^{\Imc}$ such that, for all $p,q\in\Delta^{\UO}$, $a\in\Ind(\Amc)$,
	$r\in\NR(\Kmc)$, and $C\in\Con(\Kmc)$, the following hold:
	\begin{enumerate}[label=(\arabic*)]
		\item $\pi(a)=a^\Imc$. 
		\item $p\in C^{\UO}$ implies $\pi(p)\in C^{\Imc}$.
		\item $(p,q)\in r^{\UO}$ implies $(\pi(p),\pi(q))\in r^{\Imc}$.
		\item $(p,q)\in\ira{\UO}$ implies $(\pi(p),\pi(q))\in\ira{\Imc}$.
	\end{enumerate}
	%
	This mapping $\pi$ is defined inductively based on the structure of paths.
	
	\noindent{\textbf{Case $p = a\in\Ind(\Amc).$}} Define $\pi(a):=a^\Imc$.
	
	Hence, (1) is satisfied.
	By Lemmas~\ref{lem:ioiffuo} and \ref{lem:ioiffo} and the fact that \Imc is a model of \Kmc, 
	(2) is also fulfilled.
	(3) is satisfied by the definition of \UO based on \IOR, the definition of \IOR, the fact that relations 
	between named elements can only be enforced by assertions, and, again, by $\Imc\models\Kmc$.
	(4) is fulfilled due to Lemma~\ref{lem:tails} and the arguments given for (3).
	This establishes the induction base.
	
	\noindent{\textbf{Case $p = qsd$, $s\in\NR$.}}
	By induction, assume that $\pi$ is already defined for $q$.
	By Lemma~\ref{prop:uoels}, $d$ must then be of the form $x_D\in\NIA$.  
	By Lemma~\ref{lem:ioiffo}, $x_D\in D^{\IOR}$, and hence
	$\Tail(q)\in(\exists s.D)^{\IOR}$ by the definition of paths based on \IOR.  
	Lemma~\ref{lem:ioiffuo} implies $q\in(\exists s.D)^{\UO}$.
	By the \indHyp, $\pi(q)\in(\exists s.D)^{\Imc}$.  
		Hence there is an $e\in\Delta^\Imc$ with $(\pi(q),e)\in s^\Imc$ and $e\in D^\Imc$.
		Define $\pi(p):=e$.
	(1) and (4) are trivially satisfied by this definition.
	(2) is fulfilled because
	$p\in C^{\UO}$ iff $\Tail(p)\in C^{\IOR}$ by Lemma~\ref{lem:ioiffuo}; $\Tail(p)=x_D$;
	$x_D\in C^{\IOR }$ iff $\Kmc\models D\sqsubseteq C$ by Lemma~\ref{lem:ioiffo}; and $e\in D^\Imc$, 
	and $\Imc\models\Kmc$.
	(3) is fulfilled by the definition of \UO, the fact that $(\pi(q),e)\in s^\Imc$, and $\Imc\models\Kmc$.

	\noindent{\textbf{Case $p=q\ir d$.}}
	We assume $\pi$ is defined for $q$.
	By Lemma~\ref{prop:uoels}, $d\in\NIT$ and has the form (i) $x_{D,e}$ or (ii) $\ell_e$, where $e$ is determined by $q$.
	In case~(i), we can argue as in the previous case.
		By Lemma~\ref{lem:ioiffo}, $x_{D,e}\in D^{\IOR}$, and hence
		$\Tail(q)\in\overline{D}^{\IOR}$ by the definition of paths based on \IOR.  
		Lemma~\ref{lem:ioiffuo} implies $q\in\overline{D}^{\UO}$.
		By the \indHyp, $\pi(q)\in\overline{D}^{\Imc}$.  
			Hence there is an $e\in\Delta^\Imc$ with $(\pi(q),e)\in \ira\Imc$ and $e\in D^\Imc$.
			Define $\pi(p):=e$.
	In the case (ii), then set $\pi(p):=\pi(q)$.
		(1) and (3) are trivially satisfied by this definition.
		(4) is fulfilled by the definition of \UO, the \indHyp, and the fact that $(\pi(q),e)\in \ira\Imc$.
		(2) is fulfilled for (i) by reasons analogous to the ones given in the previous case w.r.t.\ (2).
		For (ii), we have that
		$p\in C^{\UO}$ iff $\Tail(p)\in C^{\IOR}$ by Lemma~\ref{lem:ioiffuo}; $\Tail(p)=\ell_e$;
		$\ell_e\in C^{\IOR }$ iff 
		$\Kmc\models \underline{C}(e)$ if $e\in\NI(\Kmc)$ and 
		$\Kmc\models E\sqsubseteq \underline{C}$ if $e=x_E\in\NIA$ by Lemma~\ref{lem:ioiffo}; 
		by Lemma~\ref{prop:uoels}, $\Tail(q)=e$.
		If $e\in\NI(\Kmc)$, then the \indHyp w.r.t.\ (1), $\Imc\models\Kmc$,
		$(p,q)\in \ira\UO$, and  
		the fact that (4) is fulfilled, yield $\pi(p)\in C^{\Imc}$. 
		In case $e=x_E\in\NIA$, then $p\in E^{\UO}$ holding by Lemmas~\ref{lem:ioiffo} 
		and~\ref{lem:ioiffuo},
		the \indHyp w.r.t.\ (2), $\Imc\models\Kmc$,
		$(p,q)\in \ira\UO$, and the previous observation that (4) is fulfilled, yield $\pi(p)\in C^{\Imc}$. 

Given this mapping $\pi$, we show that every homomorphism of $\Phi$ into \UO, which justifies some 
answer to $\Phi$, composed with $\pi$ yields a homomorphism of $\Phi$ into \Imc. This is an obvious
consequence of the four properties satisfied by $\pi$.
\end{proof}


\section{Proofs for Section~\ref{sec:rewriting}}
\label{sec:appB}

\veronika{explain (see comment):\\
By the assumption that the KB contains no role synonyms, there is a prime implicant for every set $R\in\NR$ for which there is an implicant. 
}

To prove Theorem~\ref{thm:iocert}, we first need to construct the query $\aquery'$ used in the definition
of the rewriting \Rew. Let \aquery be a CQ.
Consider a new binary predicate 
$\irl$ which we assume to be always interpreted by the canonical interpretation \IOR and its unraveling \UO 
as follows:
 \begin{align*}
\irla{\IOR}&:=\{(e,\ell_e)\in\Delta^{\IOR}\times\NIL\}\\
 \irla{\UO}&:=\{(p\cdot e,p\cdot e\ir  \ell_e)\in\Delta^{\UO}\times\Delta^{\UO}\}.
 \end{align*}
We construct the FO query $\aquery'$ by exhaustively applying
the unfolding rules in Figure~\ref{fig:unf:rules},
\begin{figure}[tb]	
	\fbox{		
	\resizebox{0.93\columnwidth}{!}{
		\begin{array}[t]{lrll}
	\mbox{ }\\[-0.5ex]
	(UF1) & \overline{C}(x)&\rightarrow&\exists y.\ir(x,y)\wedge C(y)\\
	(UF2) & \underline{C}(x)&\rightarrow&\exists y_1,y_2.\rho(x,y_1)\wedge\irl(y_1,y_2)\wedge C(y_2)\\
	(UF3) & C\sqcap D(x)&\rightarrow&C(x)\wedge D(x)\\
	(UF4) & \exists r.C(x)&\rightarrow&\exists y.r(x,y)\wedge C(y),r\in\NR 
	\\[-0.75ex]
	\mbox{ }
	\end{array}
}}
\caption{Unfolding rules for constructing $\aquery'$}
\label{fig:unf:rules}
\end{figure}
where a rule application corresponds to replacing a conjunction on the left-hand side of the rule, by the corresponding one
on the right-hand side.
In the rules, $C$ and $D$ denote arbitrary complex concepts, and $y_1,y_2$, and $y$ fresh variables for each rule 
application.
\veronika{maybe better move that to the figure caption?}
Notice that the terms used in the construction of $\aquery'$ are based on the original query $\aquery$, 
and hence do not apply to the existentially quantified variables introduced during the application of the unfolding 
rules in this construction.

\veronika{proof outline}
Given the CQ \aquery, let $\pi$ be a valuation of the variables in $\aquery$ such that 
$\UO\models\aquery(\pi(\vec{x}))$. 
We define the mapping $\tau:\Term(\Rew)\to\Delta^{\IOR}$ inductively on the application of the unfolding rules
from Figure~\ref{fig:unf:rules} as follows:
\begin{itemize}
\item $\tau(t)=\Tail(\pi(t))$ for all $t\in\Term(\aquery)$; 
\item if $\rho(x,y)\land C(y)$ was introduced by (UF1), then $\tau(y)=x_{C,b}$ if
$\tau(x)$ is of the form $b,x_{D,b}$, or $\ell_b$, with $b\in\Ind(\Amc)$, and 
$\tau(y)=x_{C,x_D}$ if $\tau(x)\in[x_D]_{\ira\IOR}$;
\veronika{$x_{C,b}\in\NIU$... similarly for the others, no? also below}
\rafael{don't understand the comment}
\item if $\rho(x,y_1)\land\rho_L(y_1,y_2)$ was introduced by (UF2) then
\begin{itemize}
\item $\tau(y_1)=b, \tau(y_2)=\ell_b$ if $\tau(x)$ is of the form $b,x_{D,b}$, or $\ell_b$, with $b\in\Ind(\Amc)$, and
\item $\tau(y_1)=x_C, \tau(y_2)=\ell_{x_C}$ if  $\tau(x)\in[x_D]_{\ira\IOR}$; and
\end{itemize}
\item if $r(x,y)\land C(y)$ was introduced by (UF4), then $\tau(y)=x_C$
\end{itemize}
It is easy to see that this function $\tau$ is well defined. We now show that $\IOR\models\aquery'(\tau(\vec{x}))$.
\begin{lemma}
\label{lem:psi:cq}
$\IOR\models\aquery'(\tau(\vec{x}))$.
\end{lemma}
\begin{proof}
The proof is by induction on the application of unfolding rules for constructing $\aquery'$.
Let $\aquery^0,\aquery^1,\ldots$ be the sequence queries obtained at each application of an unfolding rule, 
with $\aquery^0=\aquery$. For the base case, it follows from Lemma~\ref{lem:ioiffuo} and the construction
of $\tau$ that $\IOR\models\aquery(\tau(\vec{x}))=\aquery^0(\tau(\vec{x}))$. Suppose now that 
$\IOR\models\aquery^n(\tau(\vec{x}))$. We prove that
$\IOR\models\aquery^{n+1}(\tau(\vec{x}))$ by a case analysis over the rule applied. As a prototypical case, we
show the result only for (UF1); all other cases are analogous.
\veronika{haha... didn't you once comment that this is not sufficient ;P}
\rafael{yes; we need to extend it, but for the moment I am focusing on the rest}

\noindent{\bf (UF1)} $\aquery^{n+1}$ is obtained from $\aquery^n$ by replacing $\overline{C}(x)$ by 
$\exists y.\rho(x,y)\land C(y)$, where $x\in\Term(\aquery^n)$. By induction, we know that 
$\tau(x)\in\exists\rho.C^{\IOR}$.
By Fact~\ref{prop:iosim}, $\tau(x)$ can only be an equivalence class of the form $[b]_{\ira{\IOR}}$, 
$b\in\Ind(\Amc)$, or $[x_D]_{\ira{\IOR}}$, $x_D\in\NIA$. From Lemma~\ref{lem:ioiffo} it then follows that
$\Kmc\models\overline{C}(a)$ or $\Kmc\models D\sqsubseteq \overline{C}$, respectively.
 But then
$(\tau(x),x_{C,e})\in\rho^{\IOR}$ and $x_{C,e}\in C^{\IOR}$, where $e$ is either $b$ or $x_D$, respectively.
This implies that $\IOR\models \aquery^{n+1}(\tau(\vec{x}))$.
\end{proof}
This lemma shows that $\tau$ is an $(a_1,\ldots,a_k)$-match for $\IOR$ and $\aquery'$. 
\veronika{is this defined already?}
Since our goal is to
show that it is a match for \Rew, we need to prove that $\IOR\models\Psi_i(\tau)$ for all $i,1\le i\le 3$.
Notice that all the new variables introduced to \Rew during the rewriting are existentially quantified,
and hence cannot be answer variables; moreover, the auxiliary sets $\ForkId,\ForkNeq,\ForkH$, and $\Cyc$ used are 
defined w.r.t.\ the relation \erqr. Thus, it suffices to consider only $\tau(t)$ for $t\in\Term(\aquery)$. 
\veronika{I currently don't see why the before arguments lead to that..}
\rafael{the point is that we do not need to match any other variables}
We start by showing the following result.

\begin{lemma}
\label{lem:cl:one}
Consider $s,t\in\Term(\aquery)$ such that $s\erqrOp t$ and $\pi(s)\in{\Aux}^{\UO}$. Then
\begin{enumerate}
\item\label{lem:cl:one:a} $\pi(s)=\pi(t)$ and
\item\label{lem:cl:one:b} for all terms $s',t'$ and roles $r_1,r_2$, if $r_1(s',s),r_2(t',t)\in\aquery$, then 
						  $\pi(s')=\pi(t')$.
\end{enumerate}
\end{lemma}
\begin{proof}
By definition, \erqr is the smallest transitive and reflexive relation that includes 
$\{(t,t')\mid r_1(s,t),r_2(s',t')\in\aquery,r_1,r_2\in\NR,t\erqOp t' \}$, 
and is closed under~\eqref{cond:closure} (see page~\pageref{cond:closure}).

We prove~\ref{lem:cl:one:a} by induction on the definition of \erqr. 
If $s\erqrOp t$ with $s\not=t$, then $s\erqOp t$ and there exist $r_1(s',s),r_2(t',t)\in\Phi$. Since $\pi$ is a match for
$\aquery$ and \UO, we have that $\pi(s),\pi(t)\in\Ind(\Amc)\cup\Aux^{\UO}$. Given $\pi(s)\in\Aux^{\UO}$ and 
Fact~\ref{prop:iosim}, we get $\pi(s)=\pi(t)$. The result follows trivially for the reflexive closure.
We only need to prove it for the closure under transitivity and \eqref{cond:closure}.

Assume that the result holds for $s\erqrOp t'$ and $t'\erqrOp t$. Then, by the induction hypothesis, $\pi(s)=\pi(t')=\pi(t)$.

Suppose now that $r_1(s,s'),r_2(t,t')\in\aquery$ and the result holds for $s'\erqrOp t'$. 
Since
$(\pi(s),\pi(s'))\in r_1^{\UO}$, $\pi(s')\in\Aux^{\UO}$, and hence, by induction, $\pi(s')=\pi(t')$. But then,
by the construction of the unraveled interpretation, $\pi(s)=\pi(t)$.

The property~\ref{lem:cl:one:b} follows directly from~\ref{lem:cl:one:a} and the closure under~\eqref{cond:closure}.
\end{proof}

Using this result, we can then show that $\tau$ is a match for the auxiliary queries $\Psi_i$.

\begin{lemma}
\label{lem:psi:all}
If $\UO\models\aquery(\pi)$, then $\IOR\models\Psi_i(\tau)$ for all $i,1\le i\le 3$.
\end{lemma}
\begin{proof}
For $\Psi_1$, let first $v\in\AVar(\aquery)$. 
By definition of query answers,
$\pi(v)\in\Ind(\Amc)^{\UO}$. 
But then, $\tau(v)=\pi(v)$ by definition,
and hence
$\tau(v)\notin \Aux_\rho^{\IOR}\cup\Aux^{\IOR}$ since that set is disjoint with $\Ind(\Amc)^{\IOR}$.

Regarding the other cases, we proceed by contradiction and suppose that $\tau(v)\in\Aux^{\IOR}$. If $v\in\ForkNeq$, then
there is no implicant of $\In([v]_{\erqr})$ by the definition of \ForkNeq. For every $r\in\In([v]_{\erqr})$, there
exists $r(s_r,t_r)\in\aquery$ such that $t_r\erqrOp v$. 
Moreover, $\tau(v)\in\Aux^{\IOR}$ implies $\pi(v)\in\Aux^{\UO}$ by the interpretations of \Aux; 
\veronika{original argument does only hold for concepts?}
thus $\pi(v)=\pi(t_r)$ (Lemma~\ref{lem:cl:one}), and 
$(\pi(s_r),\pi(v))\in r^{\UO}$. Given that \UO is the unraveling of the interpretation \IOR; i.e., it is tree-shaped,
this implies that for all $r,r'\in\In([v]_{\erqr})$
$\pi(s_r)=\pi(s_{r'})$; but then every $r\in\In([v]_{\erqr})$ is an implicant of $\In([v]_{\erqr})$, 
\veronika{I am not sure of the reason since implicant is defined wrt \Rmc ie the whole KB?!\\maybe:
due to the interpretation of roles in \UO?
}
yielding
a contradiction.

Finally, if $v\in\Cyc$ then there exist
$m\geq 0$
\veronika{regarding below check also comment in paper. I am lost in this case because of the cyc def. I must have some blockade in my thoughts?}
\rafael{fixed}
 $r_i(t_i,t'_i)\in\aquery, 0\le i\le m$, and $j,0\le j\le m$, with
$(v,t_j)\in{\erqr\cup\erq}$. Since $\tau(v)\in\Aux^{\IOR}$, it follows from Lemma~\ref{lem:cl:one} 
and Corollary~\ref{lem:uosimniaels} that $\pi(t_j)\in\Aux^{\UO}$, and therefore $\pi(t_j')=\pi(t_j)\cdot r_jd$ for some
$d\in\Delta^{\IOR}$. 
In particular, $\pi(t'_j)\in\Aux^{\UO}$. Additionally, we know that
$(t'_i,t_{i+1})\in{\erqr\cup\erq}$ for all $i,0\le i<m$, and $(t'_m,t_0)\in{\erqr\cup\erq}$. Repeating this
argument, we obtain that $\pi(t_j)=\pi(t_{j+m\mod m+1})=\pi(t_j)r_jp$ for some path $p$, which is a 
contradiction.
\veronika{considering the defintion of \UO based on \IOR?}

To prove that it is a match for $\Psi_2$, let $(\{t_1,\ldots,t_k\},\zeta)\in\ForkId$ such that $t_\zeta\in\Aux^{\IOR}$.
Then, $\pi(t_\zeta)\in\Aux^{\UO}$ and there are terms $t'_1,\ldots,t'_k\in\zeta$ and role names $r_1,\ldots,r_k$
such that $r_i(t_i,t'_i)\in\aquery$ for all $i,1\le i\le k$. By Lemma~\ref{lem:cl:one} (\ref{lem:cl:one:b}),
$\pi(t_i)=\pi(t_j)$, and hence $\tau(t_i)=\tau(t_j)$
 holds for all $1\le i,j\le k$.

Finally, we prove the claim for $\Psi_3$. Let $(\I,\zeta)\in\ForkH$ such that $\tau(t_\zeta)\in\Aux^{\IOR}$. 
Since $\Pre(\zeta)\neq\emptyset$, $t^\Pre_\zeta$ is defined and 
$\Gamma:=\{r\in\NR\mid (\tau(t^\Pre_\zeta),\tau(t_\zeta))\in r^{\IOR}\}\neq\emptyset$ has an implicant $r\in\Gamma$.
\veronika{the last consequence is not clear to me}
Lemma~\ref{lem:cl:one}, together with the definition of $\tau$ then yields:
\begin{itemize}
\item $\tau(t)=\tau(t_\zeta)$ for all $t\in\zeta$, and
\item $\tau(t)=\tau(t^\Pre_\zeta)$ for all $t\in\Pre(\zeta)$.
\end{itemize}
Let $\Psi:=\{s\in\NR\mid s(t,t')\in\aquery$ for some $t\in\Pre(\zeta),t'\in\zeta\}$. Then $\Psi\subseteq\Gamma$
 and 
hence $r$ is an implicant for $\Psi$; moreover, there exists a prime implicant $\hat r\in \Gamma$ of $\Psi$. 
\veronika{all this only holds because of the 2 above items, right?}
\rafael{yes}
Then we have $(\tau(t^\Pre_\zeta),\tau(t_\zeta))\in \hat{r}^{\IOR}$ and 
$\hat r\in\I$.
\end{proof}

The following is a direct consequence of Lemmas~\ref{lem:psi:cq} and~\ref{lem:psi:all}.

\begin{corollary}
\label{cor:nec}
Let $\aquery$ be a CQ. If $\UO\models\aquery(a_1,\ldots,a_k)$, then 
$\IOR\models\Rew(a_1,\ldots,a_k)$.
\end{corollary}

To finish the proof of Theorem~\ref{thm:iocert}, we need to show that the converse implication holds too; that is, that our filter conditions fit their purpose of sifting out spurious answers. We proceed similarly as before and consider an arbitrary, but fixed, match $\pi$ for $\IOR$ and \Rew. 
In order to define a corresponding match $\tau$ for \UO and \aquery, we have to find the relevant domain elements in \UO.
\veronika{. where i deg used...}
 The filter conditions are helpful there.
%
%
In the proof, we use the \emph{degree} of an equivalence class of \erqr.
Intuitively, this is the largest length of a `sequence' (modulo \erqr) of role atoms in \aquery starting in an element of the class.
 Formally, the degree of the
equivalence class $\zeta$, written $d(\zeta)$, is the largest $n\ge 0$ such that there exists a sequence 
$r_1(t_0, t'_1),r_2(t_1,t'_2),\ldots, r_n ( t_{n-1} , t'_n ) \in\aquery$
with $t_0 \in \zeta$, and $r_i\in\NR$ and $t'_i\erqrOp t_i$ for all $i,1\le i < n$. 
If no such largest 
natural number exists, then define $d(\zeta):=\infty$.
\begin{lemma}~
\label{lem:claimA1}
\begin{enumerate}
\item If $\pi(t) \in \Aux^{\IOR}$, then $d([t]_{\erqr}) <\infty$.
\item If $s\erqrOp t$ and $\pi(s)\in\Aux^{\IOR}$,then
  \begin{enumerate}[label=(\roman{*}), ref=(\roman{*})]
  \item $\pi(s) = \pi(t);$ \label{lem:A1:one}
  \item If $r_1(s', s), r_2(t', t)\in\aquery$, $r_1,r_2\in\NR$, then $\pi(s') = \pi(t')$. \label{lem:A1:two}
  \end{enumerate}
\end{enumerate}
\end{lemma}
\begin{proof}
To prove the first point, suppose that 
$d([t]_{\erqr})=\infty$. Since $\pi(t)\in\Aux^{\IOR}$, $t$ cannot be an answer variable, and hence
$t\in\QVar(\aquery)$. Since $\aquery$ is finite,
$d([t]_{\erqr})=\infty$ implies that $t\in\Cyc$. 
\veronika{really? the very last relation in the \Cyc condition is \erq. what if we have an \erqr relation?}
But, then, $\Psi_1$ contains the conjunct $\neg \Aux(t)$,
which contradicts the given fact that $\pi(t)\in\Aux^{\IOR}$.

Consider now the second point. Since $\pi(s)\in\Aux^{\IOR}$, then by the previous point we know that
$d([s]_{\erqr})<\infty$. We prove~\ref{lem:A1:one} by induction on the degree of $[s]_{\erqr}$.
If $d([s]_{\erqr})=0$, then, since $s\erqrOp t$, it follows that $s\erqOp t$. 
\veronika{because? (i currently wonder why)}
Additionally, if $s\not=t$, then 
there must exist $r_1(s',s),r_2(t',t)\in\aquery$ with $r_1,r_2\in\NR$. In particular, this means that
$t\in\Aux^{\IOR}$ 
\veronika{because? similar arguments as in below comment?}
and hence, by Fact~\ref{prop:iosim}, $\pi(s)=\pi(t)$.
For the induction step, we label the construction of \erqr by defining
\begin{align*}
\erqr^{(0)} :={} & \{(t, t)\mid t\in\Term(\aquery)\} \cup {} \\
		    & \{(t,t')\mid r_1(s,t),r_2(s',t'){\in}\aquery,r_1,r_2{\in}\NR,t\erqOp t' \},
\end{align*}
\veronika{the first just ensure that all terms are present, right? since r1,r2 don't have to be different, the second also holds for terms occuring in 1 role atom only- see paper comment}
and
\begin{align*}		    
\erqr^{(i+1)} {:=} &  \erqr^{(i)}\cup {} \\
&\{(s,t)\mid \exists s'.s\;\erqr^{(i)}\; s'\text{ and }s'\;\erqr^{(i)}\; t\}\ \cup\\ 
&\{(s,t){\mid}\exists r_1(s,s'),r_2(t,t'){\in}\aquery,r_1,r_2{\in}\NR,s'\erqr^{(i)}t'\}.
\end{align*}
It is easy to see that $\erqr=\bigcup_{n\ge 0}\erqr^{(i)}$. We show by induction on $i$ that, if
$s\;\erqr^{(i)}\;t$, $d([s]_{\erqr})=n$, and $\pi(s)\in\Aux^{\IOR}$, then $\pi(s)=\pi(t)$. 
The induction base, for $i=0$ is trivial. 
\veronika{really? ok, $s\;\erqr^{(0)}\;t$ implies $s=t$; or $s\;\erq\;t$, then Fact 3 together with the role interpretation in \IOR (only \NIA elements can be successors apart from named ones) yields $\pi(s)=\pi(t)$. }
For the induction step, we consider two cases.

\noindent{\bf[Case 1]} If there is an $s'$ such that $s\;\erqr^{(i)}\; s'\;\erqr^{(i)}\;t$, then, by the induction on $i$, we know
	that $\pi(s)=\pi(s')$, and hence $\pi(s')\in\Aux^{\IOR}$; moreover, $s'\in[s]_{\erqr}$, which implies that
	$d([s']_{\erqr})=n$. By the induction hypotheses, we similarly derive that $\pi(t)=\pi(s')$, yielding $\pi(t)=\pi(s)$.

\noindent{\bf[Case 2]} 
\veronika{i find it a bit confusing that the prime versions are switched here}
\rafael{right: will fix it}
If there exist $r_1(s,s'),r_2(t,t')\in\aquery$ with $s'\;\erqr^{(i)}\;t'$, then, since $\pi(s)\in\Aux^{\IOR}$,
\veronika{$\pi(s')\in\Aux^{\IOR}$? looking at the following, it seems that the above is just a confusion? you mean $r_1(s',s),r_2(t',t)\in\aquery$?}
	it follows that $\pi(s')\in\Aux^{\IOR}$. 
	\veronika{by?}
	Moreover, $d([s']_{\erqr})<d([s]_{\erqr})$. By induction on the degree,
	we have that $\pi(s')=\pi(t')$. Since $\pi$ is a match for $\Psi_2$, it follows that $\pi(s)=\pi(t)$.
	
The proof of~\ref{lem:A1:two} follows immediately from~\ref{lem:A1:one} and the fact that $\pi$ is a match for
\Rew.
\veronika{especially because of $\Psi_2$, no?}
\rafael{yes}
\end{proof}

Recall that we constructed the query $\aquery'$ by applying the unfolding rules to the original query $\aquery$.
This query $\aquery'$ satisfies some useful properties, which later support us to find the above mentioned relevant elements in \UO, too.

\begin{lemma}\label{lemma:claim:A2} 
The unfolding $\aquery'$ of $\aquery$ satisfies the following properties:
\begin{enumerate}[label=(\alph{*}), ref=(\alph{*})]
\item For every $v\in\Var(\aquery')\setminus\Var(\aquery)$ there is at most one atom 
	$r(t,v)\in\Rew$ with $r\in\NR\cup\{\irl\}$ and $t\in\Term(\Rew)$.
\item For all $v\in\Var(\aquery)$, if $r(t,v)\in\Rew$, $r\in\NR\cup\{\irl\}$, and $t\in\Term(\Rew)$, then $r(t,v)\in\aquery$.
\item If there is a sequence 
$r_0(t_0,t'_0),\ldots,r_m(t_m,t'_m)\in\Rew$ with $m\geq 0$, 
$t'_i \erqrOp t_{i+1}$ or $t'_i \erqOp t_{i+1}$,  
for all $i<m$, and $t'_m \erqOp t_0$, then 
$t_i,t'_i\not\in\Var(\aquery')\setminus\Var(\aquery)$ and, in particular, $r_i\not=\irl$ for all $r_i$.
\end{enumerate}
\end{lemma}
\begin{proof}
Each unfolding step uses a freshly introduced variable as successor in an atom 
$r\in\NR\cup\{\irl\}$ that is introduced in the same step and that it does not use other variables as successors.
This directly implies~(a) and (b).
Together with the fact that the unfolding only uses fresh variables as successors (i.e., also in \ir-atoms), the assumption that a predicate 
$\irl\in\Rew$ can only have been introduced during unfolding yields (c).
\end{proof}

\veronika{What is the purpose of the following? what does the relation capture? used to identify the terms that are mapped to the same element in \UO?}
We now define the relation $\sim_\pi$ to be the reflexive and transitive closure of the following relation on
$\Term(\aquery)$:
\begin{align*}
& \{(s,t)\mid s\erqrOp t,\pi(s),\pi(t)\in \Aux^{\IOR}\} \cup {}\\
& \{(s, t)\mid r_1(s, s'), r_2(t, t')\in\Phi,\pi(s')\in\Aux^{\IOR},s'\erqrOp t'\}. 
\end{align*}
Clearly, $\sim_\pi$ is an equivalence relation.
From Lemma~\ref{lem:claimA1} it follows that  if {$s\sim_\pi t$}, then $\pi(s)=\pi(t)$.
Consider now the query $\Psi$ obtained from $\aquery'$ by identifying all terms $t,t'\in\Term(\aquery)$ where
$t\sim_\pi t'$. It is easy to see that $\pi$ is also a match for this query $\Psi$.

We can now prove the following proposition.
As the previous two lemmas, it supports us in finding those  elements in \UO that can be used to answer $\Phi$.
\veronika{why is this a proposition and the others are lemmas?}
\rafael{no specific reason, just because it is big}
\veronika{stopped here}
\begin{proposition}
\label{prop:three:point}
\begin{enumerate}[label=(\Roman{*}), ref=(\Roman{*})]
\item 
If $v\in\QVar(\Psi)$ and $ \pi(v)\in \Aux^{\IOR}$, then there is at most one $t\in\Term(\Psi)$ 
such that $r(t,v)\in\Psi$ for some $r\in\NR\cup\{\irl\}$;
\item  
If $v\in\QVar(\Psi)$, $\pi(v)\in \Aux^{\IOR}$, and $t\in\Term(\Psi)$ is 
such that $\Gamma=\{r\mid r(t,v)\in\Psi\}\not=\emptyset$, then there is an implicant $s$ for 
$\Gamma$ with $(\pi(t), \pi(v))\in s^{\IOR}$;
\item 
If $r_0(t_0,t'_0),\ldots,r_m(t_m,t'_m)\in\Psi$ with $m\geq 0$, $r_i\in\NR\cup\{\irl\}$, $t'_i \erqOp t_{i+1}$
for all $i<m$, and $t'_m \erqOp t_0$, 
then $\pi(t_i),\pi(t'_i)\not\in \Aux^{\IOR}$ for all $i\leq m$.
\end{enumerate}
\end{proposition}
\begin{proof}
\noindent{\bf (I)} Let $\pi(v)\in\Aux^{\IOR}$ and suppose that there exist $r_1(t_1,v),r_2(t_2,v)\in\Psi$ with
$r_1\not=r_2$. By Lemma~\ref{lemma:claim:A2} (a) we know that $v\in\Var(\aquery)$. From 
Lemma~\ref{lemma:claim:A2} (b) it follows that there are $r_1(s_1,s'_1),r_2(s_2,s'_2)\in\aquery$ s.t.\
$s_1\sim_\pi t$, $s_2\sim_\pi t$, and $s_1\sim_\pi v\sim_\pi s_2$. But then, $\pi(s_1)=\pi(v)$. Then,
$t_1\sim_\pi t_2$ and hence $t_1=t_2$.

\noindent{\bf (II)} Lemma~\ref{lemma:claim:A2} implies the existence of such implicant for all variables
introduced during unfolding. Let now $v\in\Var(\aquery)$ such that $\pi(v)\in\Aux^{\IOR}$ and 
$\Gamma\not=\emptyset$. Since $\pi$ is a match for $\psi_1\land\psi_3$, there exists an implicant
$s$ for $\In([v]_{\erqr})$ with $(\pi(t^\Pre_{[v]}), \pi(t_{[v]}))\in s^{\IOR}$. Moreover, 
we have
$t^\Pre_{[v]} \sim_\pi t_{[v]}$ and $t_{[v]}\sim_\pi v$. Hence $\pi(t_{[v]}) = \pi(t)$ and $\pi(t_{[v]}) = \pi(v)$.
Thus, $s$ is the required implicant for $\Gamma$.

\noindent{\bf (III)}
Let $r_0(t_0,t'_0),\ldots,r_m(t_m,t'_m)\in\Psi$ with $m\geq 0$, $r_i\in\NR\cup\{\irl\}$, 
$t'_i \erqOp t_{i+1}$, for all $i<m$, and $t'_m \erqOp t_0$. Since unfolding does not replace any variables, there
must exist $r_0(s_0, s'_0), \ldots, r_{m}(s_{m}, s'_{m})\in\aquery$ with $s_i\sim_\pi t_i$ and $s'_i \sim_\pi t'_{i}$ and 
$s'_i \erqOp s_{i+1}$, for all $i<m$, and $s'_m \erqOp s_0$.
Assume that $\pi(t'_i)\in \Aux^{\IOR}$ for some $i\le m$. Then $\pi(s_i)=\pi(t_i)$, and thus $\pi(s_i)\in\Aux^{\IOR}$
and $s_i\in\QVar(\aquery)$. But then, $s_i\in\Cyc$, and thus $\neg\Aux(s_i)$ appears in $\aquery'$, yielding a 
contradiction.
\end{proof}

We now define a mapping $\tau:\Term(\Psi)\to\Delta^{\UO}$ such that for every two terms
$t,v\in\Term(\Psi)$ it holds that $\Tail(\tau(t))=\pi(t)$ and if $(\pi(t),\pi(v))\in\ira{\IOR}$, then 
$(\tau(t),\tau(v))\in\ira{\UO}$. 
This mapping is defined recursively, depending on the properties of the term $t$.
\begin{enumerate}
\item Let $t\in\Term(\Psi)$ be such that $\pi(t)\not\in\Aux\cup\Aux_\ir$. Then define
$\tau(t):=\pi(t)$. In particular, this defines $\tau(t)$ for all 
$t\in\AVar(\Psi)\cup(\Term(\Psi)\cap\NI)$.
\item Let $v\in\QVar(\Psi)$ be such that $\pi(v)\in\Aux^{\IOR}$ and there 
is neither an atom $r(t, v)\in\phi_L$, 
$r\in\NR\cup\{\irl\}$, 
nor a symbol $t\in\Term(\Psi)$ with $v\erqrOp t$ and $v\not=t$ 
 (i.e., there is no atom $\ir(v,t')\in\aquery$ or $\ir(t',v)\in\aquery$, $t'\in\Term(\Psi)$).
By the definition of \UO and since each 
$d\in\Delta^{\IOR}$ is reachable from an element of $\Ind(\Amc)^{\IOR}$, 
there are sequences $d_0,\ldots,d_n\in\Delta^{\IOR}$ and 
$r_0,\ldots,r_{n-1}\in\NR\cup\{\ir\}$ such that 
$d_0\in\Ind(\Amc)^{\IOR}$,$d_n=\pi(v)$, $(d_i,d_{i+1})\in r^{\IOR}$ if 
$r\in\NR$, and 
$(d_i,d_{i+1})\in\irO$ if 
$r=\ir$ for all $0\leq i<n$.
Then define $\tau(v):=d_0r_0d_1\cdots r_{n-1}d_n\in\Delta^{\UO}$.
\item Let $v\in\QVar(\Psi)$ with $|[v]_{\erqr}|>1$, be such that there is no
$t\in\Term(\Psi)$ with $(v,t)\in\erqr$ for which $\tau(t)$ is already defined nor exists an
atom $r(t', t)\in\aquery$, $r\in\NR\cup\{\irl\}$, $t'\in\Term(\Psi)$. $\tau(v)$ is then defined as in 
the previous item.
\item If $\tau(v)$ is undefined and there is an atom $r(t,v)\in\Psi$ with $r\in\NR$ and $\tau(t)$ defined, then by
property (II) of Proposition~\ref{prop:three:point} there is an implicant $s$ for 
$$\Gamma=\{r\mid r(t,v)\in\Psi\}\neq\emptyset$$ 
such that $(\pi(t),\pi(v))\in s^{\IOR}$.
In this case, we define $\tau(v):=\tau(t)\cdot s\pi(v)$. 
Since $\Tail(\tau(t))=\pi(t)$ and $(\pi(t),\pi(v))\in s^{\IOR}$, we have $\tau(v)\in\Delta^{\UO}$.
\item If $\tau(v)$ is undefined and there exists a symbol $t\in\Term(\Psi)$ with $v\erqrOp t$ and $\tau(t)$ defined,
then
\begin{enumerate}
 \item
  If $\pi(t)\in[a]_{\ira{\IOR}}$, $a\in\Ind(\Amc)$, set $\tau(v)$ to an arbitrary element 
  $p\in\Delta^{\UO}$ with $\Tail(p)=\pi(v)$. 
 \item
  If $\pi(t)\in[x_C]_{\ira{\IOR}}$, $x_C\in\NIA$, by construction $\Tail(\tau(t))=\pi(t)$.
  By Fact~\ref{prop:iosim} and Proposition~\ref{prop:uoels}, $\tau(t)$ must be of the form 
  $\tau(t)=p\cdot rx_C$, $\tau(t)=p\cdot rx_C(\ir x_{D'_i,x_C})^i$, or 
  $\tau(t)=p\cdot rx_C\ir a_{x_C}$ for some $r\in\NR$, $p\in\Delta^{\UO}$, and $i\geq 1$.
  If $\pi(v)=x_C$ set $\tau(v):=p\cdot rx_C$. If $\pi(v)$ is of the form $\pi(v)=x_{E,x_C}$, then 
  there is an element $p'\cdot x_C(\ir x_{E'_j,x_C})^j\ir x_{E,x_C}\in\Delta^{\UO}$, $j\geq 0$
  (Proposition~\ref{prop:uoels}). But then, we also have the element 
  $e=p\cdot rx_C(\ir x_{E'_j,x_C})^j\ir x_{E,x_C}\in\Delta^{\UO}$, and
  can set $\tau(v):=e$. The case for $\pi(v)=a_{x_C}$ is analogous to the previous case.
 \end{enumerate}
 \item If $\tau(v)$ is undefined and there is an atom $\irl(t,v)\in\Psi$ with $\tau(t)$ defined, 
 then set $\tau(v):=\tau(t)\cdot\ir\pi(v)$.
 Since $\Tail(\tau(t))=\pi(t)$ and $(\pi(t),\pi(v))\in \irla{\IOR}$, we have $\tau(v)\in\Delta^{\UO}$. 
\end{enumerate}

We first show that this mapping is well defined. For the first two cases, this is clearly the case.
The third case is only applicable once for every equivalence class of \erq by construction, and hence
$\tau(v)$ is also well defined. 
By the property (I) of Proposition~\ref{prop:three:point}, the term $t$ used for defining $\tau(v)$ in the fourth
case is unique, which implies that this case is well defined too.
Consider now the fifth case. We must show that if there exist several terms $t$ for which $\tau$ is already 
defined, the equivalence class chosen for $\tau(v)$ is the same for all of them. If there is any such term
$t$ such that $\pi(t)\in\Ind(\Amc)$, then this is obviously the case. Otherwise, $\tau(t)$ must have been defined
in one of the steps 3 to 6. Step 3 can only be used to define $\tau(t)$ for one $t$ in each equivalence class.
Afterwards, all other members of this class are mapped, by step 4, to the same element $\tau(t)$.
By Lemma~\ref{lemma:claim:A2}~(b), steps 4 and 6 can only be applied once, and only if step 3 was not applied
before to the same term. The last step is well defined because all atoms of the form $\irl(t,v)\in\Psi$
are introduced at the construction of $\aquery'$, which always introduces new successor variables. If this 
step is applicable then the step 4 is not applicable. 
Overall, this means that the mapping $\tau$ is unambiguously defined; i.e., each term can only be mapped to
one element of $\Delta^{\UO}$.

It remains to be shown that $\tau(t)$ is defined for all terms $t\in\Term(\Psi)$. This follows from property
(III) of Proposition~\ref{prop:three:point}, which states that there cannot exist a cycle in $\Psi$ where a 
variable is mapped to an unnamed element.

\begin{lemma}
The mapping $\tau$ is a match for \UO and $\Psi$. 
\end{lemma}
\begin{proof}
To show this result, it suffices to consider only concepts of the form $A\in\NC$, thanks to the 
properties of the unfolding rules.
It is immediate that $\UO\models A(\tau(t))$ for all $A(t) \in \Psi$, since $\Tail(\tau(t)) = \pi(t)$, 
which is a property of the construction of $\tau$, 
and Lemma~\ref{lem:ioiffuo}. 

Let now $r(t,t')\in\Psi$, for some $r\in\NR$. If 
$\pi(t),\pi(t')\not\in\Aux^{\IOR}\cup\Aux_\ir^{\IOR}$, then 
$\tau(t)=\pi(t),\tau(t')=\pi(t')$, and $(\pi(t),\pi(t'))\in r^{\UO}$ must hold by the 
definition of \UO.
If $\pi(t')\in\Aux^{\IOR}$, then the construction of $\tau$ implies that 
$\tau(t')=\tau(t)\cdot s\pi(t)$ with $\Tmc\models s\sqsubseteq r$. By the 
definition of \UO, it then follows that $(\tau(t),\tau(t'))\in r^{\UO}$. 
The cases that $\pi(t)\in\Aux^{\IOR}\cup\Aux_\ir^{\IOR}$ and $\pi(t')\in\Ind(\Amc)$, and $\pi(t')\in\Aux_\ir^{\IOR}$
cannot occur, by the manner in which \IOR is constructed. 
For $\ir(t,t')\in\Psi$, $(\pi(t),\pi(t'))\in\ira{\IOR}$, given by the semantics, directly yields that 
$(\tau(t),\tau(t'))\in\ira{\UO}$ since this is a property of the construction of $\tau$. 
For $\irl(t,t')\in\Psi$, we have $(\pi(t),\pi(t'))\in\irla{\IOR}$ and that $\pi(t)$ and $\pi(t')$ must be of the form 
$e$ and $a_e$, $e\in\Ind(\Amc)\cup(\NIA\cap\Delta^{\IOR})$.
But then, the construction of $\tau$ implies that there is an element 
$\tau(t')=\tau(t)\cdot\ir\pi(t')\in\Delta^{\UO}$, and then the definition of $\irla{\UO}$ yields
$(\tau(t),\tau(t'))\in\irla{\UO}$. 
\end{proof}

Finally, we adapt $\tau$ to get a mapping from $\Term(\aquery)$ to $\Delta^{\UO}$ 
by setting
$\tau(t) := \tau(t')$ if $t \in\Term(\aquery) \setminus\Term(\Psi)$ and $t\erqrOp\pi(t')$. 
It is a simple task to verify that $\tau$ is a match for \UO and $\aquery$. 
Since $\tau(t) = \pi(t)$ if $\pi(t)\in\Ind(\Amc)$ for all $t\in\Term(\Psi)$, it is also clear that $\tau$ is an 
$(a_1,\ldots, a_k)$-match. 
Overall, what this means is that every match for \Rew in \IOR is also a match for \aquery in \UO.

\begin{corollary}
\label{cor:suff}
If $\IOR\models\Rew(a_1,\ldots,a_k)$, then $\UO\models\aquery(a_1,\ldots,a_k)$.
\end{corollary}

Corollaries \ref{cor:nec} and \ref{cor:suff} imply that 
$\IOR\models\Rew(a_1,\ldots,a_k)$ if and only if $\UO\models\aquery(a_1,\ldots,a_k)$. By 
Lemma~\ref{lem:uo-cert}, the latter is the case iff $(a_1,\ldots,a_k)\in \Cert(\aquery,\Kmc)$, which finishes the
proof of Theorem~\ref{thm:iocert}.

\end{tr}


\end{document}
